\documentclass[draft,12pt,reqno,a4paper]{amsart}
\usepackage[margin=3cm,footskip=1cm]{geometry}

\usepackage{amssymb} \usepackage{amsmath} \usepackage{amsfonts}
\usepackage{amsthm} \usepackage{mathtools}
\usepackage{enumerate} \usepackage[latin1]{inputenc}
\usepackage[shortlabels]{enumitem}

\usepackage{color} 
\usepackage{graphicx} \usepackage{verbatim}

\DeclareMathOperator*{\slim}{s--lim}

 \newcommand{\cs}{{\rm the Cauchy-Schwarz inequality }}

\newcommand{\N}{{\mathbb{N}}} 
\newcommand{\R}{{\mathbb{R}}}

 \renewcommand{\c}{{\rm c}}
\newcommand{\e}{{\rm e}} \newcommand{\ess}{{\rm ess}}
 \renewcommand{\i}{{\rm i}}
\renewcommand{\d}{{\rm d}}

\renewcommand{\Re}{{\rm Re}\,}

\DeclarePairedDelimiter\inp\langle\rangle
\newcommand\parb[2][]{#1 \big ( #2#1\big )}

 \renewcommand{\exp}{{\rm exp}}
 \newcommand{\mfor}{\text{ for }}

 \newcommand{\vB}{{\mathcal B}}

 \newcommand{\vH}{{\mathcal H}}

\theoremstyle{plain}%default
\newtheorem{thm}{Theorem}[section]
\newtheorem{proposition}[thm]{Proposition}
\newtheorem{lemma}[thm]{Lemma} \newtheorem{corollary}[thm]{Corollary}
\theoremstyle{definition} 
 \newtheorem{example}[thm]{Example}
\newtheorem{examples}[thm]{Examples}
 \newtheorem{cond}[thm]{Condition}
 \newtheorem{remark}[thm]{Remark}
\newtheorem{remarks}[thm]{Remarks}
 \newtheorem*{remarks*}{Remarks}
\newtheorem*{remark*}{Remark}

\numberwithin{equation}{section}

\title[Stationary scattering theory on mainifolds, I]{Stationary scattering theory on mainifolds, I}

\thanks{K.I. is supported by JSPS KAKENHI, grant nr. JP25800073, 17K05325. 
E.S. is supported by the Danish Council for Independent Research $|$ Natural Sciences, grant nr. DFF-4181-00042. \\
\quad This manuscript is published in J. Funct.\ Anal.\ with title ``Radiation condition bounds on manifolds with ends''}

\author{K. Ito}
\address[K. Ito]{Graduate School of Mathematical Sciences, The University of Tokyo\\
3-8-1 Komaba, Meguro-ku, Tokyo 153-8914, Japan}
\email{ito@ms.u-tokyo.ac.jp}

\author{E. Skibsted} \address[E. Skibsted]{Institut for Matematiske
  Fag \\
  Aarhus Universitet\\ Ny Munkegade, 8000 Aarhus C, Denmark}
\email{skibsted@math.au.dk}

\begin{document}
\begin{abstract} 
We study spectral theory for the 
Schr\"odinger operator on  manifolds possessing   an  escape
  function. A particular class of examples are manifolds  with
  Euclidean and/or hyperbolic ends. Certain exterior domains  for  possibly
unbounded obstacles are included. 
We prove Rellich's theorem, the limiting absorption principle,
radiation condition bounds and the Sommerfeld uniqueness result,
striving to extending and refining
previously known spectral results on
manifolds.
The proofs are given by an extensive use of
commutator arguments. These arguments have a classical spirit (essentially) not
involving energy cutoffs or microlocal analysis and require, 
presumably,  minimum regularity and decay properties of perturbations.
This paper has interest of its own right, but it also serves as a
basis for  the stationary scattering theory 
 developed fully in the sequel \cite{IS3}.
\end{abstract}

\maketitle
\tableofcontents

\section{Introduction}\label{sec:intr} 

Let $(M,g)$ be a connected Riemannian manifold.
In this paper we study the spectral theory for the geometric Schr\"odinger operator
\begin{align*}
H=H_0+V; \quad
H_0=-\tfrac12\Delta=\tfrac12p_i^*g^{ij}p_j,\ p_i=-\mathrm i\partial_i,
\end{align*}
on the Hilbert space ${\mathcal H}=L^2(M)$. 
In the last expression we formally move to local coordinates.
(Here and henceforth we do not assume the existence of a global frame
and we tacitly move back and forth to local coordinates.)
The potential $V$ is
real-valued and bounded, and the self-adjointness of $H$ is realized
by the Dirichlet boundary condition. 
Assuming a kind of \textit{end} structure on $M$,  
we prove in this paper Rellich's theorem, the limiting absorption
principle, radiation condition bounds and 
the Sommerfeld uniqueness result.
 Our assumptions are general enough to cover
for example  manifolds with  finitely many ends of
  (mixed) Euclidean and hyperbolic
types studied recently by Kumura  \cite{K4}. 
Another example is a `scattering manifold' as  introduced by Melrose \cite{Me}.
Our theory  also covers  certain exterior domains  for  possibly
unbounded and non-smooth obstacles in a manifold. 
For the Euclidean
model  certain 
unbounded regular exterior domains were studied previously by Constantin \cite{Co}, Minskii
\cite {Min} and Il'in \cite{Il1, Il2, Il3}. 

 To prove the above results we invent a commutator argument with some weight inside.
This commutator argument has a classical spirit, to some extent
resembling \cite {Min, Sa, K4} (see \cite { K4} for a
  more extensive list of references). In particular we are not going to
use Mourre theory \cite{Mo, JP, Do}. Rather our  `conjugate operator'
$A$ is the generator of a semigroup of  (a kind of) radial translations 
and not a group of radial dilations   which for a limited class
  of metrics fits into standard Mourre theory \cite{ Do},
and the commutator includes an appropriate weight
depending on the context (see Lemmas
    \ref{lem:14.10.4.1.17ffaabb}, \ref{lem:14.10.4.1.17ffaabbqqq},
    \ref {lem:14.10.4.1.17ffaa} and \ref{lem:14.10.22.18.8ff} for such
     contexts). For a
    recent comprehensive study of related issues although with a
    different focus we refer to \cite{RT}. This study
    is restricted to an asymptotically conic manifold (i.e. 
    Euclidean type)  with a 
short-range potential. Although commutator arguments also play an
important role in  \cite{RT} they appear rather different from ours.

Our paper extensively 
employs explicit commutator computations of differential operators and
to a limited degree tools from
  functional analysis (primarily semigroup theory). Similarly microlocal analysis
  is virtually absent in this paper.  As an advantage our
assumptions on regularity and decay properties of perturbations
appear rather weak.

Based on the results of  this paper,  
in \cite{IS3}
we develop fully  the stationary scattering theory in
 a similar  but somewhat  more restrictive framework, 
and  in particular we provide  
a complete characterization of asymptotics for appropriate generalized eigenfunctions
at infinity. The radiation condition bounds of the
  present paper (see
  Corollary \ref{thm:radiation-conditions}) appear to a large extent
  as new results in a  manifold setting and will  play a crucial role
  for this study.

\subsection{Setting and results}\label{subsec:result}

\subsubsection{Basic setting}

We shall  study manifolds for which there exist ends in a somewhat disguised form.
\begin{cond}\label{cond:12.6.2.21.13}
Let $(M,g)$ be a connected Riemannian manifold  of dimension $d\ge 1$.
There exist a function $r\in C^\infty(M)$ with image $r(M)=[1,\infty)$
and constants 
$c>0$ and $r_0\ge 2$ such that:
\begin{enumerate}
\item\label{item:12.4.7.19.40} 
The gradient vector field $\omega=\mathop{\mathrm{grad}} r\in\mathfrak X(M)$ is
forward complete in the sense that the integral curve of $\omega$ is defined
for any initial point $x\in M$ and any non-negative time parameter $t\ge 0$.
\item\label{item:12.4.7.19.40b}
The bound $|\mathrm dr|=|\omega|\ge c$ holds on $\{x\in M\,|\, r(x)> r_0/2\}$. 
\end{enumerate}
\end{cond}

We call each component of the open subset $E=\{x\in M\,|\, r(x)>r_0\}$ 
an \emph{end} of $M$,
and the function $r$ may model a distance function there.
The last interpretation is supported by a part of \eqref{eq:13.9.28.13.28} below too.
The set  $E$ is obviously the continuous disjoint union of $r$-spheres 
\begin{align*}
S_R=\{x\in M\,|\,r(x)=R\};\quad R> r_0,
\end{align*} 
which are submanifolds of $M$ due to 
(\ref{item:12.4.7.19.40b}) of Condition~\ref{cond:12.6.2.21.13} 
and the implicit function theorem.
Then we can canonically construct the \textit{spherical coordinates} on $E$
along the vector field $\omega$,
however, since these coordinates are not used in this paper,
we do not give  the construction here. Note that spherical coordinates  will be important in our sequel paper \cite{IS3}.

Let us impose more assumptions on the geometry of $E$ in terms of the radius function $r$.
Choose $\chi\in C^\infty(\mathbb{R})$ such that 
\begin{align}
\chi(t)
=\left\{\begin{array}{ll}
1 &\mbox{ for } t \le 1, \\
0 &\mbox{ for } t \ge 2,
\end{array}
\right.
\quad
\chi'\le 0,
\label{eq:14.1.7.23.24}
\end{align}
and set 
\begin{align*}
\eta=1-\chi(2r/r_0),\quad
\tilde\eta =\eta |\mathrm dr|^{-2}
=\bigl(1-\chi(2r/r_0)\bigr)|\mathrm dr|^{-2}
.
%\label{eq:13.9.23.5.54b}
\end{align*}
We introduce the `radial differential operator':
\begin{align}
A
=\mathop{\mathrm{Re}} p^r
=\tfrac12\bigl( p^r+( p^r)^*\bigr);\quad 
p^r=-\mathrm i \nabla^r,\ 
\nabla^r=\nabla_{\omega}=g^{ij}(\nabla_i r)\nabla_j,
\label{eq:13.9.23.2.24}
\end{align}
and also the tensor $\ell$ 
and the associated differential operator $L$: 
\begin{align}
\ell=g-\tilde\eta\,\mathrm dr\otimes \mathrm dr,\quad 
L=p_i^*\ell^{ij}p_j
.
\label{eq:13.9.23.5.54}
\end{align}
As we can see with ease,
the tensor $\ell$ may be identified with the pull-back of $g$
to the $r$-spheres, and $L$ with the spherical part of $-\Delta$.
We remark that the tensor $\ell$ clearly satisfies 
\begin{align}
0\le \ell\le g,\quad
\ell^{\bullet i}(\nabla r)_i=(1-\eta)\mathrm dr,
\label{eq:14.12.23.13.49}
\end{align}
where the first bounds of \eqref{eq:14.12.23.13.49} are
understood as quadratic form estimates on fibers
of the tangent bundle of $M$.
The quantities  of \eqref{eq:13.9.23.5.54} will play a
major role in this paper.

Recall a local expression of the Levi--Civita connection $\nabla$:
If we denote the Christoffel symbol by 
$\Gamma^k_{ij}=\tfrac12g^{kl}(\partial_i g_{lj}+\partial_j g_{li}-\partial_lg_{ij})$,
then for any smooth function $f$ on $M$
\begin{align}
(\nabla f)_i&=(\nabla_if)=(\mathrm d f)_i=\partial_if,\quad
(\nabla^2 f)_{ij}
=\partial_i\partial_j f-\Gamma^k_{ij}\partial_kf.
\label{eq:14.3.5.0.14}
\end{align}
Note that $\nabla^2f $ is the geometric Hessian of $f$.

\begin{cond}\label{cond:12.6.2.21.13a}
There exist constants $\sigma,\tau,C>0$ such that globally on $M$
\begin{subequations}
\begin{align}
r\Bigl(\nabla^2r-
\tfrac12 \eta |\mathrm dr|^{-4}(\nabla^r |\mathrm dr|^2)\mathrm dr\otimes\mathrm dr\Bigr)
\ge 
\tfrac12\sigma |\mathrm dr|^2\ell-Cr^{-\tau} g
\label{eq:13.9.5.3.30}
\end{align}
as quadratic forms on fibers of the tangent bundle of $M$, 
and
\begin{align}
\begin{split}
|\mathrm dr|^2&\leq C,\quad
\bigl|\nabla^r|\mathrm dr|^2\bigr|\leq Cr^{-1-\tau/2}
,\quad 
\Delta r\le C,\quad
\bigl|\ell^{\bullet i}\nabla_i\Delta r\bigr|\leq Cr^{-1-\tau/2}.
\end{split}\label{eq:13.9.28.13.28}
\end{align}
\end{subequations}
\end{cond}

Condition~\ref{cond:12.6.2.21.13a} 
says that the ends are geometrically growing. 
For any $R>r_0$ we let
$\iota_R\colon S_R\hookrightarrow M$ 
be the inclusion.
In case where $r$ is an exact distance function, i.e.\ $|\mathrm dr|=1$ on $E$,
the Hessian $\nabla^2r$ has no radial components in the spherical coordinates,
and $(\nabla^2r)_{|S_R}$ can be identified with the pull-back $\iota_R^*(\nabla^2r)$,
which is exactly the second fundamental form of $S_R$,
and $(\Delta r)_{|S_R}$ with the mean curvature $\mathop{\mathrm{tr}}[\iota_R^*(\nabla^2 r)]$ 
of $S_R$. 
In general under Condition~\ref{cond:12.6.2.21.13a} 
the radial components of $\nabla^2r$ do not necessarily vanish, 
but we may still somehow regard the quantity
$$\nabla^2r
-\tfrac12 \eta |\mathrm dr|^{-4}(\nabla^r |\mathrm dr|^2)\mathrm dr\otimes\mathrm dr$$ 
as the second fundamental form with negligible error, cf.\ Lemma~\ref{lem:151209},
and hence the bound \eqref{eq:13.9.5.3.30}  implies  that
the ends are growing, bounding the minimum  curvature
of $S_R$ below: 
For any $\sigma'\in (0,\sigma)$ there exists $R_{\sigma'}\ge r_0$ such that 
for all $R\ge R_{\sigma'}$
\begin{align*}
R\iota_R^*(\nabla^2r)
=R\iota_R^*\Bigl(\nabla^2r
-\tfrac12 \eta |\mathrm dr|^{-4}(\nabla^r |\mathrm dr|^2)\mathrm dr\otimes\mathrm dr\Bigr)
\ge \tfrac12\sigma'|\mathrm dr|^2\iota_R^*g.
\end{align*}
The bounds in (\ref{eq:13.9.28.13.28})
together with \eqref{eq:13.9.5.3.30}, 
see also Lemma~\ref{lem:151209}, are 
connected to the regularity properties of the mean curvature of $S_R$,
and, in particular, we can  bound the maximum curvature above, 
since $\iota_R^*(\nabla^2r)$ is strictly positive for $R>r_0$ large enough. 
We also remark that in agreement with conventions one could reasonably call the radius
  function $r$ an {\it escape function}, 
  since as a consequence of the convexity property \eqref{eq:13.9.5.3.30}
  the complete geodesics with $r$ globally large enough are non-trapped.
One benefit of our 
indirect description of the geometry of $(M,g)$ is 
that it is obviously stable under small perturbations.
When it is difficult to compute an exact distance function,
we may choose a more useful distance-like function to verify the conditions.

Finally we impose a long-range type condition on the
potential $V$.
More precisely, taking into account a metric quantity related to the volume growth
of the ends, we formulate it in terms of an \textit{effective potential} $q$
defined by
\begin{align}
q=
V+\tfrac18\tilde\eta\bigl[(\Delta r)^2+2\nabla^r\Delta r\bigr].
\label{eq:14.12.10.22.48}
\end{align}
This  quantity naturally appears.
In fact, using \eqref{eq:13.9.23.5.54} and the expressions
\begin{align}
A&=p^r-\tfrac{\mathrm i}2(\Delta r)
=(p^r)^* +\tfrac{\mathrm i}2(\Delta r),
\label{eq:12.6.27.0.43}
\end{align}
we can rewrite the Schr\"odinger operator $H$ in the form
\begin{align}
\begin{split}
H=\tfrac 12A\tilde\eta A+\tfrac 12 L
+q+\tfrac14(\nabla^r\tilde\eta)(\Delta r)
.
\end{split}\label{eq:15091010}
\end{align}

\begin{cond}\label{cond:10.6.1.16.24}
The potential $V$ is a real-valued function belonging to $L^\infty(M)$.
Moreover, there exists a splitting by real-valued functions:
  \begin{align*}
  q=q_1+q_2;\quad q_1\in C^1(M)\cap L^\infty(M),\  q_2\in L^\infty(M),
  \end{align*}
such that for some $\rho',C>0$ the following bounds hold globally on $M$:
\begin{align}\label{eq:60}
\nabla^r q_1\le Cr^{-1-\rho'},\quad
|q_2|\le Cr^{-1-\rho'}.
\end{align}
\end{cond}

A setting similar to Conditions~\ref{cond:12.6.2.21.13}--\ref{cond:10.6.1.16.24}
is used in \cite{IS2}. See also \cite{K2,K3,IS1}.
In Subsection~\ref{sec:13.9.5.0.34} 
we shall discuss concrete models of manifolds satisfying 
Conditions~\ref{cond:12.6.2.21.13}--\ref{cond:10.6.1.16.24}
along with the additional Conditions~\ref{cond:12.6.2.21.13b} and
\ref{cond:12.6.2.21.13bbb} stated below.
The models include manifolds with asymptotically Euclidean and/or hyperbolic ends,
and their conic regions. Some more general regions with unbounded obstacles
are included too.
We remark that in this  paper only derivatives of $r$ of order at most four
are used quantitatively. Throughout our presentation we use the
convention that $c$ is used for a `small'  positive 
constant while $C$ is used for a `big'  positive constant, however their
particular values  not being important.  
On the other hand the parameters $\sigma$, $\tau$ and $\rho$ appearing 
in Condition~\ref{cond:12.6.2.21.13bbb} 
are intimately related to scattering
properties of quantum particles on the model manifold. This will be
demonstrated  in \cite{IS3}.

Now let us explain  the self-adjoint realizations of $H$ and $H_0$.
Since $(M,g)$ can be incomplete, the operators $H$ and $H_0$ are not necessarily 
essentially self-adjoint on $C^\infty_{\mathrm{c}}(M)$.
We realize $H_0$ as a self-adjoint operator by imposing the 
Dirichlet boundary condition, i.e.\ 
$H_0$ is the unique self-adjoint
operator associated with the closure of the quadratic form
\begin{align*}
\langle H_0\rangle_\psi=\langle \psi, -\tfrac12 \Delta \psi\rangle,\quad \psi
\in C^\infty_{\mathrm{c}}(M).
\end{align*}
We denote the form closure and the self-adjoint realization by the same symbol $H_0$. 
Define the associated Sobolev spaces $\mathcal H^s$ by
\begin{align}\mathcal H^s=(H_0+1)^{-s/2}{\mathcal H}, \quad s\in\mathbb{R}.
\label{eq:13.9.5.1.2}
\end{align}
Then $H_0$ may be understood as a
closed quadratic form on $Q(H_0)=\mathcal H^1$.
Equivalently, $H_0$ makes sense also as a bounded operator 
$\mathcal H^1\to\mathcal H^{-1}$, whose action coincides with 
that for distributions.  
By the definition of the Friedrichs extension 
the self-adjoint realization of $H_0$ 
is the restriction of such distributional
$H_0\colon \mathcal H^1\to\mathcal H^{-1}$ to the domain:
\begin{align*} 
\mathcal D(H_0)=\{\psi\in\mathcal H^1\,|\, H_0\psi\in \mathcal H\}\subseteq \mathcal H.
\end{align*}
Since $V$ is real-valued and bounded  by Condition~\ref{cond:10.6.1.16.24}, 
we can realize the self-adjoint operator $H=H_0+V$ simply as
\begin{align*}
H=H_0+V,\quad \mathcal D(H)=\mathcal D(H_0).
\end{align*}

In contrast to (\ref{eq:13.9.5.1.2}) 
we introduce the weighted Hilbert space 
$\mathcal H_s$ for $s\in\mathbb R$ by 
\begin{align*}
\mathcal H_s=r^{-s}{\mathcal H}.
\end{align*}
We also denote the locally $L^2$-space by 
\begin{align*}
\mathcal H_{\mathrm{loc}}=L^2_{\mathrm{loc}}(M).
\end{align*}
We consider the $r$-balls $B_R=\{ r(x)<R\}$ and  
the characteristic functions 
\begin{align}
\begin{split}
F_\nu=F(B_{R_{\nu+1}}\setminus B_{R_\nu}),\ R_\nu=2^{\nu},\ \nu\ge 0,
\end{split}\label{eq:13.9.12.17.34}
\end{align}
 where $F(\Omega)$ is used for sharp characteristic
 function of a subset $\Omega\subseteq M$. 
Define the associated Besov spaces $B$ and $B^*$ by
\begin{align}
\begin{split}
B&=\{\psi\in \mathcal H_{\mathrm{loc}}\,|\,\|\psi\|_{B}<\infty\},\quad 
\|\psi\|_{B}=\sum_{\nu=0}^\infty R_\nu^{1/2}
\|F_\nu\psi\|_{{\mathcal H}},\\
B^*&=\{\psi\in \mathcal H_{\mathrm{loc}}\,|\, \|\psi\|_{B^*}<\infty\},\quad 
\|\psi\|_{B^*}=\sup_{\nu\ge 0}R_\nu^{-1/2}\|F_\nu\psi\|_{{\mathcal H}},
\end{split}\label{eq:13.9.7.15.11}
\end{align}
respectively.
We also define $B^*_0$ to be the closure of $C^\infty_{\mathrm c}(M)$ in $B^*$. 
These spaces are sometimes called
Agmon--H\"ormander spaces, due to the fact that 
the  $B$--$B^*$ framework was first applied to the limiting absorption principle 
in  \cite{AH,Ag2}; 
see also \cite[Chapter X\hspace{-.1em}I\hspace{-.1em}V]{Ho2}. However
they appeared earlier \cite{L2} in the  context of Rellich's theorem. 
By noting $R_\nu\approx r$ on $\mathop{\mathrm{supp}}F_\nu$ and the H\"older inequality 
for $\ell^p$ spaces
it is rather easy to see that for any $s>1/2$
\begin{align}\label{eq:19080223}
\mathcal H_s\subsetneq B\subsetneq \mathcal H_{1/2}
\subsetneq\vH
\subsetneq\mathcal H_{-1/2}\subsetneq B^*_0\subsetneq B^*\subsetneq \mathcal H_{-s}.
\end{align}

Using the function $\chi\in C^\infty(\mathbb R)$ of \eqref{eq:14.1.7.23.24},
define $\chi_n,\bar\chi_n,\chi_{m,n}\in C^\infty(M)$ for $n,m\ge 0$ by
\begin{align}
\chi_n=\chi(r/R_n),\quad \bar\chi_n=1-\chi_n,\quad
\chi_{m,n}=\bar\chi_m\chi_n.
\label{eq:11.7.11.5.14}
\end{align}  
 Let us introduce an auxiliary space:
\begin{align*}
\mathcal N=\{\psi\in \mathcal H_{\mathrm{loc}}\,|\, \chi_n\psi\in \mathcal
H^1\mbox{ for all }n\ge 0\}.
\end{align*}
This is a space of  functions that satisfy the
  Dirichlet boundary condition,
possibly with infinite $\mathcal H^1$-norm on $M$.
 Note that under Conditions~\ref{cond:12.6.2.21.13}--\ref{cond:10.6.1.16.24} the manifold $M$
  may be, e.g.\ a half-space in the Euclidean space, and there could be 
a `boundary' even for large $r$, which is `invisible' from inside $M$. 

\subsubsection{Rellich's theorem}
Our first theorem is Rellich's theorem, the absence of $B^*_0$-eigenfunctions 
with eigenvalues above a certain `critical energy'
$\lambda_0\in\mathbb R$ given by 
\begin{align}
\lambda_0
=\limsup_{r\to\infty}q_1
=
\lim_{R\to\infty}
\bigl(\sup\{q_1(x)\,|\, r(x)\ge R\}\bigr)<\infty
.
\label{eq:13.9.30.6.8}
\end{align}
For the Euclidean and the hyperbolic spaces and many other examples 
the critical energy $\lambda_0$ can be computed 
explicitly, see Subsection~\ref{sec:13.9.5.0.34}, and the essential
spectrum $\sigma_\ess(H)=[\lambda_0,\infty)$. 
The latter is usually seen in terms of Weyl sequences, see  \cite{K1}.
\begin{thm}\label{thm:13.6.20.0.10}
Suppose Conditions~\ref{cond:12.6.2.21.13}--\ref{cond:10.6.1.16.24},
and let $\lambda>\lambda_0$.
If a function  $\phi\in  \mathcal H_{\mathrm{loc}}$ satisfies that
\begin{enumerate}
\item\label{item:13.7.29.0.27}
$(H-\lambda)\phi=0$ in the distributional sense, 
\item\label{item:13.7.29.0.26}
$\bar\chi_m\phi\in \mathcal N \cap B_0^*$ for all $m\ge 0$ large enough,
\end{enumerate}
then $\phi=0$ in $M$.
\end{thm}

This is  known as Rellich's theorem on the Euclidean space,
and we extend it to manifolds. 
The space $B_0^*$ is optimal in the sense that (under more restrictive conditions)
we can construct  generalized eigenfunctions belonging to $B^*$, cf.\ \eqref{eq:19080223}
and \cite{IS3}. 
As far as the authors know, the  type of proof for Rellich's theorem 
given in Section~\ref{sec:absence} 
may be considered as new even for  the Euclidean space although the
result in this case  and for some class of potentials already
appears in \cite{Ka}.
In view of the weighted spaces in \eqref{eq:19080223} absence of $L^2$-eigenfunctions
discussed e.g.\ in \cite{IS2}
is clearly  weaker than Theorem~\ref{thm:13.6.20.0.10}.
Let us state it as a corollary.

\begin{corollary}\label{cor:15060822}
The operator $H$ has no eigenvalues above $\lambda_0$:
$\sigma_{\mathrm{pp}}(H)\cap (\lambda_0,\infty)=\emptyset$.
\end{corollary}

Note that in order to verify (\ref{item:13.7.29.0.26}) it suffices to have
$\bar\chi_m\phi\in \mathcal N\cap B_0^*$ for a single value of
$m$. For any function $\phi$ obeying the conditions of
Theorem~\ref{thm:13.6.20.0.10} we have $\chi_{m,n}\phi\in\mathcal
D(H)$ for all $m$ large enough and $n>m$.  See the discussion on our
self-adjoint realization of $H$ above.  We can drop the space
$\mathcal N$ if the $r$-annuli $B_{R_{\nu+1}}\setminus B_{R_\nu}$ are
relatively compact in $M$ for all large $\nu\ge 0$.  Note that an
$r$-ball $B_R$, $R\in \mathbb R$, may be unbounded under
Conditions~\ref{cond:12.6.2.21.13}--\ref{cond:10.6.1.16.24}.  If on
the other hand $M$ is complete and $B_R$ is bounded it follows from
the Hopf--Rinow theorem \cite[Theorem 1.4.8]{Jost} that $B_R$ is
relatively compact. 
%Corollary~\ref{cor:15060822}, a direct consequence
%of Theorem~\ref{thm:13.6.20.0.10}, was proven in a somewhat similar
%setting in \cite{IS2}. For a little more precise
%  comparison with \cite{IS2}, see 
%  Remark~\ref{rem:16020515} \eqref{item:160208}.

\subsubsection{Limiting absorption principle and radiation condition}

Next we discuss the limiting absorption principle
and the radiation condition
related to the resolvent 
\begin{align*}
R(z)=(H-z)^{-1}.
\end{align*}
We first establish a locally uniform bound for the resolvent $R(z)$ as
a map: $B\to B^*$.
Let us impose a compactness condition.
\begin{cond}\label{cond:12.6.2.21.13b}
In addition to Conditions~\ref{cond:12.6.2.21.13}--\ref{cond:10.6.1.16.24},  
there exists an open subset $\mathcal I\subseteq (\lambda_0,\infty)$
such that for any $n\ge 0$ and  compact interval  $I\subseteq \mathcal I$ 
the mapping 
\begin{align*}
  \chi_nP_H(I)\colon \mathcal H\to\mathcal H
\end{align*}
is  compact, where $P_H(I)$ denotes the spectral projection onto $I$ for $H$. 
\end{cond}

Due to Rellich's compact embedding theorem \cite[Theorem X\hspace{-.1em}I\hspace{-.1em}I\hspace{-.1em}I.65]{RS},
`boundedness' of $r$-balls provides a criterion for Condition~\ref{cond:12.6.2.21.13b}:
If $M$ is complete and each $r$-ball $B_R$, $R\geq 1$, is bounded there, 
then Condition~\ref{cond:12.6.2.21.13b} is satisfied for $\mathcal I=(\lambda_0,\infty)$.
More generally, even if $M$ is incomplete, 
it suffices that each $r$-ball $B_R$, $R\geq 1$, is 
isometric to a bounded subset of a complete manifold.
Condition~\ref{cond:12.6.2.21.13b} in fact includes even more general situations 
where $M$ has several ends possibly  with different
critical energies and where $r$-balls are unbounded,
cf.\ \cite{K4}. We shall discuss an example in Subsection~\ref{sec:13.9.5.0.34}.

For  notational simplicity we set for a large $C>0$ 
\begin{align}
  \begin{split}
  h&=\nabla^2r-
\tfrac12 \eta |\mathrm dr|^{-4}(\nabla^r |\mathrm dr|^2)\mathrm dr\otimes\mathrm dr
+2Cr^{-1-\tau}g\\
&\ge\tfrac12\sigma r^{-1}|\mathrm dr|^2 \ell+Cr^{-1-\tau}g\ge 0,  
  \end{split}
\label{eq:15091112}
\end{align}
cf.\ \eqref{eq:13.9.5.3.30}.
We may consider $h$, as well as $\nabla^2r$, 
as the second fundamental form with negligible error. 
For any subset $I\subseteq \mathcal I$ let us denote
\begin{align*}
I_\pm=\{z=\lambda\pm \mathrm i\Gamma\in \mathbb C\,|\,\lambda\in I,\ \Gamma\in (0,1)\},
\end{align*}
respectively. 
We also use the notation 
$\langle T\rangle_\psi=\langle\psi,T\psi\rangle$.

\begin{thm}\label{thm:12.7.2.7.9}
Suppose Condition~\ref{cond:12.6.2.21.13b}
and let $I\subseteq \mathcal I$ be a compact interval.
Then there exists $C>0$ such that 
for any $\phi=R(z)\psi$ with $z\in I_\pm$ and $\psi\in B$
\begin{align}
\|\phi\|_{B^*}+\|p^r\phi\|_{B^*}
+\langle p_i^*h^{ij}p_j\rangle_\phi^{1/2}
+\|H_0\phi\|_{B^*}
\le C\|\psi\|_B.
\label{eq:13.8.22.4.59c}
\end{align}
\end{thm}

Theorem~\ref{thm:12.7.2.7.9} is  sharper than \cite{K4} in the sense that 
we are working in the $B$--$B^*$ framework  and
that our framework is more general.
Note that it is well known that the standard Mourre commutator argument fails on
the hyperbolic manifold (see \cite { K4} for references),
but our commutator method provides a unified approach for manifolds
with   Euclidean or  hyperbolic ends (or  a mixture).  Kumura's
more classically flavoured method shares the same unified feature.

Absence of singular continuous spectrum is a standard application of
the uniform boundedness of $R(z)$ in an appropriate operator space.
 It is stated as follows.

\begin{corollary}\label{cor:15060823}
The operator $H$ has no singular continuous spectrum on $\mathcal I$:
$\sigma_{\mathrm{sc}}(H)\cap \mathcal I=\emptyset$.
\end{corollary}

 Corollary~\ref{cor:15060823} applies in  particular to the Dirichlet Laplacian on
 $\R^d\setminus K$ for any compact set $K$, which to our knowledge was
 first  proved in \cite{DaSi}.

The Besov boundedness \eqref{eq:13.8.22.4.59c} 
does not immediately imply the limiting absorption principle. Before
showing it we establish  radiation condition bounds under 
 an additional (minor) regularity and decay condition. The bounds will be
 crucial for our application \cite{IS3}.

\begin{cond}\label{cond:12.6.2.21.13bbb}
In addition to Condition~\ref{cond:12.6.2.21.13b} 
with the same $\tau>0$ appearing there,
there exist splittings $q_1=q_{11}+q_{12}$ and 
$q_2=q_{21}+q_{22}$ 
by real-valued functions
\begin{align*}
q_{11}\in C^2(M)\cap L^\infty(M),\quad
q_{12}, 
q_{21}\in C^1(M)\cap L^\infty(M),\quad
q_{22}\in L^\infty(M)
\end{align*}
and constants $\rho,C>0$
such that for $k=0,1$
\begin{align*}
%  \label{eq:pot}
\bigl|\ell^{\bullet i}\nabla_i|\mathrm dr|^2\bigr|&\le Cr^{-1-\tau/2}&
|\nabla^rq_{11}|&\le Cr^{-(1+\rho/2)/2},&
|\ell^{\bullet i}\nabla_i q_{11}|&\le Cr^{-1-\rho/2},\\
|\mathrm d\nabla^rq_{11}|&\le Cr^{-1-\rho/2},&
|\mathrm dq_{12}|&\le Cr^{-1-\rho/2},&
|(\nabla^r)^k q_{21}|&\le Cr^{-k-\rho},\\
q_{21}\nabla^r q_{11}&\leq Cr^{-1-\rho},&
|q_{22}|&\le Cr^{-1-\rho/2}.&&
\end{align*}
\end{cond}

Our  radiation condition bounds are
stated in terms of the radial derivative $A$ defined in 
\eqref{eq:13.9.23.2.24} 
and an asymptotic complex phase $a$ given below.
Pick a
  smooth decreasing function $r_\lambda\geq r_0$ of
  $\lambda>\lambda_0$ such that
  \begin{align}
    \lambda+\lambda_0-2q_1\geq 0\mfor r\geq r_\lambda/2,
  \label{eq:14.6.29.23.46}
  \end{align}
  and that  $r_\lambda= r_0$ for all $\lambda$ large enough.
Then we set for $z=\lambda \pm \i \Gamma \in\mathcal I\cup\mathcal I_\pm$
\begin{align}
a=a_z=\eta_\lambda\Bigl[|\mathrm dr|\sqrt{2(z-q_1)}\pm\tfrac14(p^rq_{11})\big/(z-q_1)\Bigr];\quad
\eta_\lambda=1-\chi(2r/r_\lambda),\label{eq:13.9.5.7.23}
\end{align}
respectively, where the branch of square root is chosen such that 
$\mathop{\mathrm{Re}}\sqrt w>0$ for $w\in \mathbb C\setminus (-\infty,0]$.
Note that the phase $a=a_\pm$ of \eqref{eq:13.9.5.7.23} 
is an approximate solution to the radial Riccati equation
\begin{align}
\pm p^ra+a^2-2|\mathrm dr|^2(z-q_1)=0
\label{eq:15.3.11.19.35}
\end{align} 
in the sense that it makes the quantity on the left-hand side of 
\eqref{eq:15.3.11.19.35} small for large $r\ge 1$.
The first term in the brackets of \eqref{eq:13.9.5.7.23} alone already gives
an approximate solution to the same equation, however with the 
second term a better approximation is obtained, cf.\ Lemma~\ref{lem:13.9.2.7.18}
and Remark~\ref{rem:15.3.12.14.9}.
Let 
\begin{align}
\beta_c
=\tfrac12\min\{\sigma,\tau,\rho\}>0.\label{eq:10c}
\end{align}

\begin{thm}
  \label{prop:radiation-conditions} 
  Suppose Condition~\ref{cond:12.6.2.21.13bbb} and let $I\subseteq \mathcal I$
be a compact interval.
  Then for all  $\beta\in [0,\beta_c)$
  there exists $C>0$ such that 
  for any $\phi=R(z)\psi$ with $\psi\in r^{-\beta}B$ and $z\in I_\pm$
\begin{align}
\|r^\beta(A\mp a)\phi\|_{B^*}
+\langle p_i^*r^{2\beta}h^{ij}p_j\rangle_{\phi}^{1/2}
&\leq C\|r^\beta\psi\|_B,\label{eq:14cccCff}
\end{align} 
respectively.
\end{thm}

As an application we obtain the limiting absorption principle.

\begin{corollary}\label{thm:12.7.2.7.9b}
Suppose Condition ~\ref{cond:12.6.2.21.13bbb},
and let $I\subseteq \mathcal I$ be a compact interval.
For any $s>1/2$
and $\epsilon\in (0,\min\{s-1/2,\beta_c,(2+\rho)/4\})$ 
there exists $C>0$ such that for $k=0,1$ and any $z,z'\in I_+$ or $z,z'\in I_-$ 
\begin{align}
\|p^k R(z)-p^k R(z')\|_{\mathcal B(\mathcal H_s,\mathcal H_{-s})}\le C|z-z'|^{\min\{\epsilon,1\}}.
\label{eq:14.12.30.21.52}
\end{align}
In particular, the operators $p^k R(z)$, $k=0,1$, attain
uniform limits as $I_\pm \ni z \to \lambda \in I$ in the norm topology of 
${\mathcal B}(\mathcal H_s,\mathcal H_{-s})$, say denoted by 
\begin{align}
p^k R(\lambda\pm\mathrm i0)=\lim_{I_\pm \ni z\to \lambda}p^k R(z), \quad \lambda\in I,
\label{eq:14.12.30.21.53}
\end{align}
respectively. 
These limits $p^k R(\lambda\pm\mathrm i0)\in{\mathcal B}(B,B^*)$,
and $R(\lambda\pm\mathrm i0)\colon B\to \mathcal N\cap B^*$.
\end{corollary}

Now we have the limiting resolvents $R(\lambda\pm\mathrm i0)$.
The radiation condition bounds for real spectral parameters
follow directly from Theorem~\ref{prop:radiation-conditions}.

\begin{corollary}
  \label{thm:radiation-conditions}
  Suppose Condition~\ref{cond:12.6.2.21.13bbb}  and let 
  $I\subseteq \mathcal I$ be a compact interval.
  Then for all $\beta \in [0,\beta_c)$
  there exists $C>0$ such that 
  for any $\phi=R(\lambda\pm\mathrm i0)\psi$ with $\psi\in r^{-\beta}B$ and 
  $\lambda\in I$ 
\begin{align}
\|r^\beta(A\mp a_\pm)\phi\|_{B^*}
+\langle p_i^*r^{2\beta}h^{ij}p_j\rangle_{\phi}^{1/2}
&\leq C\|r^\beta\psi\|_B,\label{eq:14cccCa} 
\end{align}
respectively.
\end{corollary}

  We shall see in Subsection~\ref{sec:13.9.5.0.34} that for the
  Euclidean and the hyperbolic spaces without potential $V$ we have
  $\beta_c\ge 1$. Hence in
  these cases the bound 
   \eqref{eq:14cccCa} hold for any
   $\beta\in [0,1)$. 
We remark that for the Euclidean space  and a sufficiently regular
   potential the bound \eqref{eq:14cccCa} is well known for  $\beta\in [0,1)$, cf.\ 
   \cite{Is,Sa,HS}.  However in this case one can actually
  allow $\beta\in [1,2)$, cf.\ \cite{HS}.

As another application of the radiation condition bounds 
we can characterize the limiting resolvents $R(\lambda\pm\mathrm i0)$. 
For the Euclidean space such characterization
is usually referred to as the \textit{Sommerfeld uniqueness result}, see
for example  \cite{Is}. 
\begin{corollary}\label{thm:13.9.9.8.23}
  Suppose Condition ~\ref{cond:12.6.2.21.13bbb}, and let
  $\lambda\in\mathcal I$, $\phi\in \mathcal H_{\mathrm{loc}}$ and
  $\psi\in r^{-\beta}B$ with  $\beta\in [0,\beta_c)$.
Then 
$\phi=R(\lambda\pm\mathrm i0)\psi$ holds if and only if
both of the following conditions hold:
\begin{enumerate}[(i)]
\item\label{item:13.7.29.0.29}
$(H-\lambda)\phi=\psi$ in the distributional sense.
\item\label{item:13.7.29.0.28}
$\phi\in \mathcal N\cap
r^\beta B^*$  and $(A\mp a)\phi\in r^{-\beta}B^*_0$.
\end{enumerate}
\end{corollary}

\subsection{Discussion of simple models}\label{sec:13.9.5.0.34}

Let us provide several examples. We shall tacitly ignore the fact
that possibly the range $r(M)=[a,\infty)$ with $a\neq 1$. Of course
the change $r\to r-a+1$ will then yield an $r$ conforming with our
conditions (including Condition \ref{cond:12.6.2.21.13}).
\subsubsection{{Ends of warped-product type}}
Let $(M,g)$ be a complete Riemannian manifold.
Suppose that 
there exist a relatively compact open subset $B\subseteq M$ 
and a $(d-1)$-dimensional closed Riemannian manifold $(S,h)$ 
such that isometrically
\begin{align*}
 M\setminus B \cong [2,\infty)\times S,\quad \partial B\cong
    \{2\}\times S,
\end{align*} 
and that in the coordinates $(r,\sigma)\in [2,\infty) \times S$
the metric $g$ is of warped-product type:
\begin{align}
g(r,\sigma)
=\mathrm{d}r\otimes \mathrm{d}r+f(r)h(\sigma);\quad
h(\sigma)=h_{\alpha\beta}(\sigma)\,\mathrm{d}\sigma^\alpha\otimes \mathrm{d}\sigma^\beta.
\label{eq:13.9.26.17.39}
\end{align}Here the Greek indices run over
    $2,\dots,d$. To make contact to Condition
    \ref{cond:12.6.2.21.13} we identify 
 $\{r_0\}\times S=S\subseteq M$ 
    for any fixed  $r_0\geq 4$ and modify the coordinate $r$ suitably to
    become a globally defined smooth function.
 Such a modified $r$ obviously conforms with 
the bounds in \eqref{eq:13.9.28.13.28} of 
Condition~\ref{cond:12.6.2.21.13} on a compact subset.
Below we examine in more detail in this particular setting 
the content of a number of the  bounds of Subsection~\ref{subsec:result} by specifying $f$ explicitly.
 Whence we consider $E\cong
    (r_0,\infty)\times S$, and more generally, the spherical coordinates are well-defined on
 $M\setminus B$ and  the Christoffel
  symbols are computed there as follows:
\begin{align*}
\Gamma^r_{rr}&{}=0,&
\Gamma^{r}_{r\alpha}&{}=\Gamma^{r}_{\alpha r}=0,&
\Gamma^{\alpha}_{rr}&{}=0,\\
\Gamma^r_{\alpha\beta}&{}=-\tfrac12f' h_{\alpha\beta},&
\Gamma^\alpha_{r\beta}&{}=\Gamma^\alpha_{\beta r}=\tfrac12 (f'/f) \delta^{\alpha \beta},&
\Gamma^\alpha_{\beta\gamma}&=(\Gamma_S)^\alpha_{\beta\gamma},
%\label{eq:12.8.24.3.45}
\end{align*}
where $\delta^{\alpha \beta}$ denotes Kronecker's $\delta$
and $\Gamma_S$ the Christoffel symbol for $h$ on $S$.
Hence 
\begin{align}
\begin{split}
|\mathrm dr|^2=1,\qquad
\nabla^2r=\tfrac12 f'h,\qquad
\Delta r=\tfrac{d-1}2 f'/f,\qquad
\iota^*_R\nabla^3r=0.
\end{split}\label{eq:12.9.9.10.52}
\end{align} 
For the last calculation
    $\iota^*_R\nabla^3r=0$ (to be relevant only to \cite{IS3}) we used the compatibility condition \eqref{eq:12.9.27.1.21}.
 Now we can verify the conditions of Subsection
    \ref{subsec:result} (with  $V\equiv0$) for the following examples.
\begin{examples}\label{examples:radi-cond-appl}
\begin{enumerate}
\item\label{item:13.9.27.20.52}
Let 
\begin{align*}
f(r)=r^{\theta};\quad \theta> 0.
\end{align*}
Condition~\ref{cond:12.6.2.21.13bbb} is
satisfied for $\sigma=\theta$, any $\tau>0$, $\rho'=2$ and $\rho=6$, 
and the critical energy is $\lambda_0=0$. 
The Euclidean space corresponds to $f(r)=r^2$ and $S$ being the standard unit sphere.

\item\label{item:13.9.27.20.53}
Let 
\begin{align*}
f(r)=\exp(\delta r^\theta);\quad 0<\delta,\ 
0<\theta< 1.
\end{align*}
Condition~\ref{cond:12.6.2.21.13bbb} is
satisfied for any $\sigma>0$, any $\tau>0$, $\rho'=2-2\theta$ and $\rho=6-4\theta$, 
and the critical energy is $\lambda_0=0$.

\item\label{item:13.9.27.20.54}
Let 
\begin{align*}
f(r)&=C\exp(\kappa r+\delta_\theta (r));\quad C,\kappa>0, \ \theta<1.
\end{align*} Here $\theta$ is an order parameter in
  the sense that the derivatives 
  \begin{align*}
    \delta^{(k)}_\theta (r)  
=O(r^{\theta-k});  \ k=0,1,2,\dots.
  \end{align*}
Condition~\ref{cond:12.6.2.21.13bbb} is
satisfied for any $\sigma>0$, any $\tau>0$, $\rho'=1-\theta$ and $\rho=4-2\theta$, 
and the critical energy is $\lambda_0=(d-1)^2\kappa^2/32$. 
 The hyperbolic space corresponds to $f(r)=(\sinh r )^2$
 and $S$ being the standard unit sphere, for which $\theta<1$ may be arbitrary.  
\end{enumerate}
Note that, if we can choose $2\beta_c=\min\{\sigma,\tau,\rho\}>1$, the above models also 
fulfill Condition~1.16 (2) of \cite{IS3}. 
In particular all of the results of \cite{IS3} apply to these examples.
\end{examples}

Furthermore we can  perturb the models of Examples~\ref{examples:radi-cond-appl}.
For example, we can add to (\ref{eq:13.9.26.17.39}) some lower order terms, 
whether warped-product type or not.
We can also put any compact obstacle  or attach handles topologically.
Obstacles can be non-compact if 
the gradient vector field $\omega$ is inward pointing as follows.

\begin{example}\label{ex:subset} Let $(M,g)$ be any of the Riemannian manifolds discussed above.
Let $\Omega\subseteq M$ be a domain such that its $r$-sections 
$U_R=\Omega\cap (\{R\}\times S)$ are increasing: For some $R_0\ge 3$
\begin{align*}
U_R\subseteq U_{R'}\quad\text{for all } R'\ge R\ge R_0-1
.
\end{align*}
We can modify the function $r$ on $\Omega\setminus ([R,\infty)\times S)$
so that the gradient vector field $\omega$ is forward complete on $\Omega$.
Then $\Omega$ satisfies Condition~\ref{cond:12.6.2.21.13bbb} 
with the same $\sigma,\tau, \rho',\rho$
as those of $M$,
and  Condition~1.16 (2) \cite{IS3} is  fulfilled as well.
This construction  includes solid cones in the Euclidean and the
hyperbolic spaces (for example half-spaces) 
for which one has $U_R=U_{R'}$.
\end{example}

\subsubsection{Unbounded  obstacles in $\R^2$}

We give examples of manifolds $M\subseteq \R^2$ equipped
        with the Euclidean metric and possessing  non-compact
        obstacles. Again   $V\equiv0$ and we shall refer to conditions of our
        sequel \cite{IS3}.
\begin{examples} 
        \begin{enumerate}[1)]
        \item 
         Let $[x]$ denote the integer part of $x>0$ and let
        $[x]_-=[x]$ for $x\not\in\N$ and $[n]_-=n-1$ for $n\in\N$.
         Consider  the ``saw-tooth region'' defined in
        terms of a parameter 
        $K>0$ as 
        \begin{align*}
          M=\{(x,y)\in\R^2|\,x>0,\,  y>K(x-[x]_-)/(1+[x]_-)\}.
        \end{align*}
         Define then $r\geq 1$ by the formula
         \begin{align*}
           r^2=1 +x^2+(y+K)^2.
         \end{align*} In this case $\omega$ is forward complete  and
          the other conditions
         of this paper, as well as  Condition~1.16 (2) of  \cite{IS3},  are fulfilled too.
       
\item  Consider  
\begin{align*} M=\{(x,y)\in\R^2|\, xy<1 \}.
        \end{align*} Define then in
        terms of a parameter 
        $K>2$  a function $r\geq 1$ by the formula
\begin{align*} r^2=x^2+y^2+\tfrac K2\ln \parb{(y-x)^2+2}.
        \end{align*}
  We compute at the boundary $\partial M\subseteq \R^2$ 
\begin{align*}
  \tfrac12\nabla r^2\cdot \nabla(xy)&=2-K +2K(x^2+y^2)^{-1}.
\end{align*}
This  expression is negative for $r$ big, more
precisely for  $x^2+y^2>2K/(K-2)$, and forward
completeness is fulfilled at infinity. In $M$ (as well as at $\partial M\subseteq \R^2$)
\begin{align*}
 \tfrac12\nabla r^2=\big (x+K(x-y)((y-x)^2+2)^{-1},y+K(y-x)((y-x)^2+2)^{-1}\big).
\end{align*} Using  the identity $\d r=2^{-1}r^{-1}\nabla r^2$ we
then 
obtain
\begin{align*}
|r\d r|^2&=x^2+y^2+2K^2\frac{(y-x)^2} {\parb{(y-x)^2+2}^{2}}+2K\frac{(y-x)^2} {(y-x)^2+2},
\end{align*}
and hence for any multiindex
 $\alpha$
\begin{align*}
  \nabla^\alpha \parb{|\d r|^2-1}&=\ln \parb{(y-x)^2+2}O\parb{r^{-2-|\alpha|}}.
\end{align*}
Similarly for the convexity we compute
\begin{align*}
 r \nabla^2 r=\ell +\ln \parb{(y-x)^2+2}O\parb{r^{-2}}.
\end{align*}
We can easily show 
 that the conditions of this paper are  fulfilled for any
 $\sigma,\tau, \rho<2$ (in particular for some $\sigma,\tau,
 \rho>1$). More generally we can again verify Condition~1.16 (2) of
 \cite{IS3} and hence 
  obtain the conclusions  of  \cite{IS3}.
\item  Fix
$\kappa\in (0,1)$, let $\theta:=xy^{-\kappa}$ for $y>0$ and let 
$r^2:=\kappa x^2+y^2$. Consider $M\subset \R^2$ with an end described as
\begin{align*}
  E=\{(x,y)\in\R\times \R_+|\quad r>r_0, \quad -1<\theta<1\},
\end{align*} which is  a cylinder in the variables $r$ and
$\theta$. The conditions of this paper are indeed
fulfilled, cf. \cite{Min}. However they are not met with 
$2\beta_c>1$ as required in \cite{IS3}. If on the other hand
$\kappa\geq 1$ we can let $r^2:= x^2+y^2$ for this model and indeed the conditions of this paper are  fulfilled for any
 $\sigma,\tau, \rho<2$. Whence  \cite{IS3} is applicable for
 $\kappa\geq 1$. 
 This agrees with Example  \ref{ex:subset} as well as with \cite{Co}.
\end{enumerate}
\end{examples}

\subsubsection{Multi-ends with different critical energies}\label{subsec:160129}

Here we discuss Condition~\ref{cond:12.6.2.21.13b}.
Let us consider the simplest situation:
Let $M$ be the $1$-dimensional Euclidean space $\mathbb R$,
which has exactly two ends, 
and the Schr\"odinger operator $H$ be given by 
\begin{align*}
H=-\tfrac12\tfrac{\mathrm d^2}{\mathrm dx^2}+V\quad 
\text{on }\mathcal H=L^2(\mathbb R),
\end{align*}
where $V\in C^\infty(\mathbb R)$ is equal to different constants $\lambda_0<\lambda_1$
on the two ends:
\begin{align*}
V(x)=
\left\{
\begin{array}{ll}
\lambda_0&\text{for }x\ge 1,\\
\lambda_1&\text{for }x\le -1.
\end{array}
\right.
\end{align*}
If we choose $r\in C^\infty(\mathbb R)$ such that 
\begin{align*}
r=\left\{
\begin{array}{ll}
x&\text{for }x\ge 2,\\
1&\text{for }x\le 1,
\end{array}
\right.
\end{align*}
then clearly Conditions~\ref{cond:12.6.2.21.13}--\ref{cond:10.6.1.16.24}
are satisfied with critical energy $\lambda_0$.
In this case, although the $r$-balls are unbounded, 
Condition~\ref{cond:12.6.2.21.13b} is certainly satisfied:
\begin{lemma}\label{lem:160128}
Under the above setting Condition~\ref{cond:12.6.2.21.13b} holds with 
$\mathcal I=(\lambda_0,\lambda_1)$.
\end{lemma}
\begin{proof}
Fix any $n\ge 0$
and any compact interval $I\subseteq\mathcal I$.
We let $\{\psi_k\}_{k\ge 0}\subseteq \mathcal H$ be a bounded sequence,
and set $\phi_k=\chi_nP_H(I)\psi_k$.
It is clear that the sequence $\{\phi_k\}_{k\ge 0}$ is bounded 
in the (usual) Sobolev space $H^1(\mathbb R)$.
Hence by Rellich's compact embedding theorem
and the diagonal argument
it suffices to show that 
\begin{align*}
\lim_{\nu\to\infty}\sup_k\|\check\chi_\nu\phi_k\|=0;\quad
\check\chi_\nu(x)=1-\chi(-x/R_\nu),
\end{align*}
cf.\ \eqref{eq:14.1.7.23.24}.
We choose $f\in C^\infty_0(\mathcal I)$ with $f=1$ on a neighborhood of 
$I$, and decompose
\begin{align*}
\check\chi_\nu\phi_k
=f(H)\check\chi_\nu\phi_k+(1-f(H))\check\chi_\nu\phi_k.
%\label{eq:160203}
\end{align*}
The terms on the right-hand side above can be estimated by a commutator method.
We omit detailed computations, but it is typical to estimate them
by using the Helffer--Sj\"ostrand formula as follows:
Uniformly in $k,\nu\ge 0$ 
\begin{align*}
\bigl\|(1-f(H))\check\chi_\nu\phi_k\bigr\|
=
\bigl\|\bigl[\check\chi_\nu,f(H)\bigr]P_H(I)\psi_k\bigr\|
\le C_1R_\nu^{-1};
\end{align*}
similarly, 
since we have uniformly in $k\ge 0$ and $\nu\ge 1$ 
\begin{align*}
(\sup\mathop{\mathrm{supp}}f)\bigl\|f(H)\check\chi_\nu\phi_k\bigr\|^2
&\ge 
\bigl\langle f(H)\check\chi_\nu\phi_k,Hf(H)\check\chi_\nu\phi_k\bigr\rangle
\\&
= 
\bigl\langle f(H)\check\chi_\nu\phi_k,\check\chi_{\nu-1}H\check\chi_{\nu-1}f(H)\check\chi_\nu\phi_k\bigr\rangle
\\&\phantom{{}={} }
+\bigl\langle
\bigl[f(H),\check\chi_{\nu-1}\bigr]\check\chi_\nu\phi_k,H\check\chi_{\nu-1}f(H)\check\chi_\nu\phi_k\bigr\rangle
\\&\phantom{{}={}}
+\bigl\langle f(H)\check\chi_\nu\phi_k,H\bigl[f(H),\check\chi_{\nu-1}\bigr]\check\chi_\nu\phi_k\bigr\rangle
\\&
\ge 
\lambda_1 \| \check\chi_{\nu-1}f(H)\check\chi_\nu\phi_k\|^2
-C_2R_\nu^{-1}
\\&
\ge 
\lambda_1\bigl\|f(H)\check\chi_\nu\phi_k\bigr\|^2
-C_3R_\nu^{-1},
\end{align*}
it follows that 
\begin{align*}
\bigl\|f(H)\check\chi_\nu\phi_k\bigr\|
\le C_4R_\nu^{-1/2}.
\end{align*}
Hence we are done.
\end{proof}

By Lemma~\ref{lem:160128}
the results of Subsection~\ref{subsec:result}
hold  true for $\mathcal I=(\lambda_0,\lambda_1)$.
However, here we note that we may retake $r\in C^\infty(\mathbb R)$ such that 
\begin{align*}
r=|x|\quad \text{for }|x|\ge 2.
\end{align*}
Then we have Conditions~\ref{cond:12.6.2.21.13}--\ref{cond:10.6.1.16.24}
with critical energy $\lambda_1$,
and also Condition~\ref{cond:12.6.2.21.13b} for $\mathcal I=(\lambda_1,\infty)$,
since now the $r$-balls are bounded and Rellich's compact embedding theorem applies.
Hence the only exceptional energy above $\lambda_0$ is
    the `threshold' $\lambda_1$.

The above arguments easily generalize to a manifold with several ends
possibly with different critical energies.
For such a model the results of Subsection~\ref{subsec:result}
hold  true above the minimum  critical energy 
except possibly for the other critical energies, or thresholds.
(Of course  Theorem~\ref{thm:13.6.20.0.10} holds true also at these thresholds.)

In a multi-end setting the limiting absorption
  principle above the minimum critical energy is obtained in
  \cite{K4}.  There Kumura deals with an exact distance function for
  which a strong (short-range type) condition on asymptotics for $\Delta r$ holds, cf. 
  \eqref{eq:160207}. On the other hand Kumura does not require bounds
  on the first derivative of $\Delta r$ as done for the escape
  function of this paper. Whence \cite{K4} is not directly comparable
  with ours.  However it would be possible to modify our arguments to
  cover situations with less regularity as in \cite{K4,IS2}. For
  simplicity of presentation we are not going to do this, however we
  have devoted  
  Remarks~\ref{rem:16020515}
  to some elaboration.

\subsubsection{Comparison with our previous model \cite{IS1}}

Below we extract and reformulate the essential parts of the conditions of \cite{IS1}
in a form similar to the  setting of the present paper.
These conditions  are more restrictive than those of the present paper. 
\begin{cond}\label{cond:14.5.24.7.4}
Let $(M,g)$ be a connected and complete Riemannian manifold,
and let $V\in L^\infty(M)$ be real-valued.
There exist an unbounded function $r\in C^\infty(M)$ and constants 
 $\delta,\kappa,\eta,C>0$ and $r_0\ge 2$ such that:
\begin{enumerate}
\item
The $r$-balls $B_R=\{x\in M\,|\, r(x)<R\}$, $R>0$, are relatively compact in $M$.
\item
The following relations hold for $r(x)=R> r_0/2$:
  \begin{align*}
    |\mathrm dr|=1,\quad
    R\iota_R^*(\nabla^2r) \ge \tfrac12 (1+\delta)\iota_R^*g.  
\end{align*}
\item
The following estimates hold globally on $M$ for $k=0,1$:
  \begin{align}
r\ge 1,\quad 
    |\nabla^k \Delta r|\le C r^{-1/2-k-\kappa},\quad
    |V|\le Cr^{-1-\eta}. \label{eq:34b}
\end{align}
\end{enumerate}
\end{cond}
\begin{example}
Obviously Condition~\ref{cond:14.5.24.7.4} follows  from
 Conditions~1.1--1.4 of \cite{IS1}. On the other hand
 Condition~\ref{cond:14.5.24.7.4} constitutes what is used in  the
 proofs of \cite{IS1} and consequently the results of  \cite{IS1}
 remain valid under Condition~\ref{cond:14.5.24.7.4}.
\begin{subequations}
Clearly the second bound of \eqref{eq:34b} implies the {\it metric short-range condition}
  \begin{align*}
   (\Delta r)^2 +2\nabla^r\Delta r= O(r^{-1-\kappa}),
  \end{align*}
and, in order to simplify $a$ of \eqref{eq:13.9.5.7.23}, 
let us here  propose to set $q_1=0$. Then clearly $q_2=O(r^{-1-\min\{\kappa,\eta\}}).$
  Hence
Condition~\ref{cond:14.5.24.7.4} suffices for applying this
paper. Note that we are not claiming that \cite{IS3} is applicable
without additional  conditions (not to be examined here).
\end{subequations} 
\end{example}

\section{Generators of radial translations}

\subsection{Elementary tensor analysis}\label{sec:12.10.11.16.40}
Here we fix our convention for the covariant derivatives.
We formulate and use them always in terms of local expressions, 
but for a coordinate-independent representation, see \cite[p.\ 34]{Chavel}.

We shall denote two tensors by the same symbol if they are related to  each other 
through the canonical identification $TM\cong T^*M$,
and distinguish them by super- and subscripts.
We denote $TM\cong T^*M$ by $T$ for short, and set $T^p=T^{\otimes p}$.
The covariant derivative $\nabla$ acts as 
a linear operator $\Gamma(T^p)\to \Gamma(T^{p+1})$ 
and is  defined for $t\in \Gamma(T^p)$
by 
\begin{align}
(\nabla t)_{ji_1\cdots i_p}
&=\nabla_jt_{i_1\cdots i_p}
=\partial_{j} t_{i_1\cdots i_p}
-\sum_{s=1}^p \Gamma^{k}_{ji_s}t_{i_1\cdots k\cdots  i_p}.
\label{eq:12.9.27.1.20}
\end{align}
Here $\Gamma^k_{ij}=\tfrac12g^{kl}(\partial_i g_{lj}+\partial_j g_{li}-\partial_lg_{ij})$ 
is  the Christoffel symbol and $t$ is considered as a section of the
$p$-fold cotangent bundle, and we adopt the convention that a new subscript is always added to the left as 
in \eqref{eq:12.9.27.1.20}.
By the identification $TM\cong T^*M$ it suffices to discuss an expression 
only for the subscripts.
In fact, we have the compatibility condition
\begin{align}
\nabla_i g_{jk}=\partial_ig_{jk}-\Gamma_{ij}^lg_{lk}-\Gamma_{ik}^lg_{jl}=0,
\label{eq:12.9.27.1.21}
\end{align}
and then by \eqref{eq:12.9.27.1.20} and \eqref{eq:12.9.27.1.21}
the covariant derivative can be computed for the tensors of any type. 
For example, for $t\in \Gamma(T)=\Gamma(T^1)$
\begin{align*}
(\nabla t)_j{}^i
=g^{ik}(\nabla t)_{jk}
=g^{ik}\bigl(\partial_j t_k-\Gamma_{jk}^l t_l\bigr)
=g^{ik}\bigl(\partial_j g_{kl}t^l-\Gamma_{jk}^l g_{lm}t^m\bigr)
=\partial_j t^i+\Gamma_{jk}^it^k,
\end{align*}
and this extends to the general case with ease.
The covariant derivative acts as a derivation with respect to tensor product,
i.e.\ for $t\in \Gamma(T^p)$ and $u\in \Gamma(T^q)$
\begin{align}\label{eq:11.3.22.6.18}
(\nabla (t\otimes u))_{ji_1\cdots i_{p+q}}
=(\nabla t)_{ji_1\cdots i_p}
u_{i_{p+1}\cdots i_{p+q}}
+
t_{i_1\cdots i_p}
(\nabla u)_{ji_{p+1}\cdots i_{p+q}}.
\end{align}
The formal adjoint $\nabla^*\colon\Gamma(T^{p+1})\to\Gamma(T^p)$ is 
defined to satisfy
\begin{align*}
\int 
\overline{u_{ji_1\cdots i_p}}
(\nabla t)^{ji_1\cdots i_p}
(\det g)^{1/2}\,\mathrm{d}x
=\int \overline{(\nabla^* u)_{i_1\cdots i_p}}
t^{i_1\cdots i_p}(\det g)^{1/2}\,\mathrm{d}x
\end{align*}
for $u\in \Gamma(T^{p+1})$ and $t\in \Gamma(T^p)$ 
 compactly  supported in a coordinate neighborhood.
Actually we can write it in a divergence form: For $u\in \Gamma(T^{p+1})$
\begin{align*}
(\nabla^* u)_{i_1\cdots i_p}
=-(\mathop{\mathrm{div}} u)_{i_1\cdots i_p}
=-(\nabla u)_{j}{}^j{}_{i_1\cdots i_p}
=-g^{jk}(\nabla u)_{jki_1\cdots i_p}.
\end{align*}

Finally let us give several remarks.
It is clear that for any function $f\in \Gamma(T^0)=C^\infty (M)$
the second covariant derivative $\nabla^2f=\nabla\nabla f$ is symmetric, i.e.\ 
\begin{align}
\begin{split}
(\nabla^2 f)_{ij}=(\nabla^2 f)_{ji}=\partial_i\partial_jf-\Gamma_{ij}^k\partial_kf,
\end{split}
\label{eq:11.3.22.6.20}
\end{align}
and we have expressions for the Laplace--Beltrami operator $\Delta$:
\begin{align*}
\Delta f=(\nabla^2f)_i{}^i=g^{ij}(\nabla^2f)_{ij}=\mathop{\mathrm{tr}}\nabla^2f
=\mathop{\mathrm{div}}\nabla f.
\end{align*}
We note that covariant differentiation and
  contraction are commuting operations. Whence we have, for example,
  for $t\in\Gamma(T)$ and  $u\in\Gamma(T^{p+1})$
  \begin{align}\label{eq:contract}
      \begin{split}
       \nabla_k t^j u_{ji_1\cdots i_p}
	   &=(\nabla t)_k{}^ju_{ji_1\cdots i_p}
	   +t^j (\nabla u)_{kji_1\cdots i_p},\\
       \nabla_j(\nabla t)_{i}{}^i&=(\nabla^2 t)_{ji}{}^{i}=g^{ik}(\nabla^2 t)_{jik} .
\end{split}
  \end{align} 
Let us verify various estimates that 
can be deduced from Conditions~\ref{cond:12.6.2.21.13} and \ref{cond:12.6.2.21.13a}.
\begin{lemma}\label{lem:151209}
\begin{subequations}
Suppose Conditions~\ref{cond:12.6.2.21.13} and \ref{cond:12.6.2.21.13a}.
Then one has 
\begin{align}
(\nabla|\mathrm dr|^2)^i=2(\nabla^2r)^{ij}(\nabla r)_j,
\label{eq:15120921}
\end{align}
and there exists $C>0$ such that 
\begin{align}
\begin{split}
\nabla^2r\le Cg,\quad\nabla^2r\ge -Cr^{-1-\tau/2}g,\quad
\Delta r\ge -Cr^{-1-\tau/2}.
\end{split}
\label{eq:14.12.27.3.22}
\end{align}
\end{subequations}
\end{lemma}
\begin{proof}
The formula \eqref{eq:15120921} is a consequence of the above tensor analysis.
Since the tensor $h$ of \eqref{eq:15091112} is non-negative,
we have 
\begin{align*}
h\le (\mathop{\mathrm{tr}}h)g
=\bigl(\Delta r
-\tfrac12 \eta |\mathrm dr|^{-2}(\nabla^r |\mathrm dr|^2)
+C_1dr^{-1-\tau}\bigr) g
\le C_2g,
\end{align*}
and then it follows that 
\begin{align*}
\nabla^2r
\le C_2g
+\tfrac12 \eta |\mathrm dr|^{-4}(\nabla^r |\mathrm dr|^2)\mathrm dr\otimes\mathrm dr
-C_1r^{-1-\tau} g
\le C_3g.
\end{align*}
This verifies the first bound of \eqref{eq:14.12.27.3.22}.
The second bound of \eqref{eq:14.12.27.3.22} follows from 
\eqref{eq:13.9.5.3.30}, and the third bound follows by taking the trace of the second bound. 
\end{proof}
\begin{remark*} Clearly a combination of bounds of
  Conditions~\ref{cond:12.6.2.21.13a} and \ref{cond:12.6.2.21.13bbb} 
  amounts to requiring $|\nabla|\mathrm dr|^2\bigr|\le
  Cr^{-1-\tau/2}$. Suppose in addition that for some $\sigma'>\sigma$
\begin{align}\label{eq:fund2}
  R\iota_R^*(\nabla^2r) \ge \tfrac12\sigma'|\mathrm
  dr|^2\iota_R^*g;\quad R\geq r_0.
\end{align} Then by \cs and \eqref{eq:15120921} in fact
\eqref{eq:13.9.5.3.30} holds. Whence \eqref{eq:13.9.5.3.30} and
\eqref{eq:fund2} can be considered as being essentially equivalent.
  \end{remark*}

\subsection{Semigroups of radial translations}\label{subsec:15.2.5.21.22}

In this subsection we define and discuss  semigroups  of 
\textit{unnormalized} radial translations (in \cite{IS3} we are going
to 
use  spherical coordinates defined by normalized radial translations).
Let  
\begin{align*}
 y\colon \mathcal M\to M,\ 
(t,x)\mapsto y(t,x)=\exp(t\omega)(x);\quad 
\mathcal M\subseteq \mathbb R\times M,
%\label{eq:14.12.10.6.32}
\end{align*}
be the maximal flow generated by 
the vector field $ \omega$. By \eqref{item:12.4.7.19.40} of Condition~\ref{cond:12.6.2.21.13} 
the set ${\mathcal M}$ 
contains a neighborhood of $[0,\infty)\times M$ in $\mathbb R\times M$.
Note that by definition it satisfies, in local coordinates,
\begin{align}
\partial_ty^i(t,x)=\omega^i(y(t,x))
=(\nabla r)^i(y(t,x)),\quad
y(0,x)=x.
\label{eq:12.6.5.3.13c}
\end{align}
We define the ``radial translations'' 
$T(t)\colon {\mathcal H}\to{\mathcal H}$, $t\in\mathbb R$, 
by 
\begin{align}
\begin{split}
(T(t)\psi)(x)
&=J(t,x)^{1/2}
\left({\det g(y(t,x))}\big/{\det g(x)}\right)^{1/4}\psi(y(t,x))\\
&=\exp \left(\int_0^t\tfrac12(\Delta r)(y(s,x))\,\mathrm{d}s\right)\psi(y(t,x))
\end{split}
\label{eq:12.6.7.1.10c}
\end{align}
if $(t,x)\in\mathcal M$, 
and $(T(t)\psi)(x)=0$ otherwise,
where $J(t,{}\cdot{})$ is the Jacobian of the mapping 
$y(t,{}\cdot{})\colon M\to M$.
Here let us verify the well-definedness and the equivalence of the two expressions
in \eqref{eq:12.6.7.1.10c}:
\begin{proof}[Verification of (\ref{eq:12.6.7.1.10c})]
If we set the pull-back
\begin{align}
(g^*)_{ij}(t,x)= g_{\alpha\beta}(y(t,x))[\partial_iy^\alpha(t,x)]
[\partial_jy^\beta(t,x)],
\label{eq:15.4.18.15.44}
\end{align}
then we can write 
\begin{align}
J(t,x)^2\det g(y(t,x))
=\det g^*(t,x).
\label{eq:12.10.12.8.31}
\end{align}
Since $g$ and $g^*$ are subject to the same transformation rule under change of coordinates,
so are $\det g$ and $\det g^*$,
and whence the first expression of \eqref{eq:12.6.7.1.10c} is
coordinate-invariant and well-defined. 

Next we prove the second equality of \eqref{eq:12.6.7.1.10c}.
Note that by the coordinate-invariance noted above 
we can choose specific coordinates to prove it.
Let us consistently use the Roman and the Greek indices to denote 
quantities concerning $x$ and $y=y(t,x)$, respectively.
Differentiating the expression \eqref{eq:15.4.18.15.44} and 
using the compatibility condition (\ref{eq:12.9.27.1.21})
and the equation (\ref{eq:12.6.5.3.13c}),
we can compute
\begin{align}
\begin{split}
\tfrac{\partial}{\partial t}(g^*)_{ij}
&= 
[\Gamma_{\gamma\alpha}^\delta g_{\delta\beta}
+\Gamma_{\gamma\beta}^\delta g_{\alpha\delta}](\nabla r)^\gamma 
(\partial_iy^\alpha)(\partial_jy^\beta)
\\
&\phantom{={}}
+
g_{\alpha\beta}
(\partial_\gamma (\nabla r)^\alpha)(\partial_iy^\gamma)
(\partial_jy^\beta)+
g_{\alpha\beta}
(\partial_\gamma (\nabla r)^\beta)(\partial_iy^\alpha)
(\partial_jy^\gamma)\\
&=2(\nabla^2 r)_{\alpha\beta}(\partial_iy^\alpha)
(\partial_jy^\beta).
\end{split}
\label{eq:12.10.12.8.38}
\end{align} Let us choose local 
coordinates around $x$ such that 
the matrix $\left((g^*)_{ij}\right)_{i,j}$
is diagonal, and then introduce the orthonormal basis of tangents at
$y$
\begin{align*}
  t_i=(g^*)_{ii}^{-1/2}(\partial_iy^\alpha)
  \tfrac{\partial}{\partial y^\alpha};\,i=1,\dots, d.
\end{align*}
Then by \eqref{eq:12.10.12.8.38}
\begin{align}
\sum_{i=1}^d(g^*)_{ii}^{-1}\tfrac{\partial}{\partial t}(g^*)_{ii}
=2\sum_{i=1}^d(\nabla^2 r)(t_i,t_i)=2\mathop{\mathrm{tr}}((\nabla^2
r)_{\alpha\beta})_{\alpha,\beta}=2\Delta r.
\label{eq:14.8.23.21.8}
\end{align}
On the other hand,
we have by \eqref{eq:12.10.12.8.31}
\begin{align}
\begin{split}
\tfrac{\partial}{\partial t}
\ln\left(
J(t,x)^2
\det g(y(t,x))\big/{\det g(x)}\right)
&=\tfrac{\partial}{\partial t}
\ln\Bigl(\prod_{i=1}^d(g^*)_{ii}\Bigr)
\\&=\sum_i (g^*)_{ii}^{-1}\tfrac{\partial}{\partial t}(g^*)_{ii}.
\end{split}
\label{eq:14.8.23.21.9}
\end{align}
Hence the second equality of (\ref{eq:12.6.7.1.10c}) follows by \eqref{eq:14.8.23.21.8} 
and \eqref{eq:14.8.23.21.9}.
\end{proof}

Now 
it follows by the former expression of 
\eqref{eq:12.6.7.1.10c} that for any $\psi\in\mathcal H$
\begin{align*}
\|T(t)\psi\|
=\biggl(\int_{M(t)}|\psi(x)|^2\bigl(\det g(x)\bigr)^{1/2}\,\mathrm dx\biggr)^{1/2};\quad
M(t)=y(\max\{t,0\},M),
\end{align*}
and hence $T(t)$, $t\ge 0$,
form a strongly continuous one-parameter semigroup of surjective partial isometries, 
and $T(-t)$, $t\ge 0$, form that of isometries,
being the adjoints of each other: $T(t)^*=T(-t)$.
We remark that in general
$T(t)$, $t\in \mathbb R$, do not form a group,
but, if the gradient vector field $\omega$ is both forward and 
backward complete, then they do and hence are unitary.

Next we investigate the generators $A_\pm$ of semigroups $T(\pm t)$, $t\ge 0$.
We let
\begin{align*}
\mathcal D(A_\pm)&=\bigl\{\psi\in\mathcal H\,|\, 
\lim_{t\to 0^+}(\pm\mathrm it)^{-1}(T(\pm t)\psi-\psi)\mbox{ exists in }\mathcal H\bigr\},\\
A_\pm\psi&=\lim_{t\to 0^+}(\pm\mathrm it)^{-1}(T(\pm t)\psi-\psi)\quad 
\mbox{for }\psi\in\mathcal D(A_\pm),
\end{align*}
respectively. 
 By the Hille--Yosida theorem \cite[Theorem X.47a]{RS} the operators $A_\pm$ are densely defined closed operators on 
$\mathcal H$. We note that $T(-t)$ preserves $C^\infty_\c(M)$ 
(see the proof of Lemma \ref{lem:12.6.4.13.14} below) and
whence by  \cite[Theorem X.49]{RS} this space is a core for $A_-$. In
particular $A_-$ is symmetric.  In addition we  easily verify that
$A_-\subseteq A_+^*$ and therefore, cf.\ \cite[Theorem X.47a]{RS},
\begin{align}
A_\pm= A_\mp^*,
\label{eq:15.2.8.3.10}
\end{align}
respectively. 
Moreover we have inclusions
\begin{align}
C^\infty_{\mathrm c}(M)\subseteq\mathcal D(H)\subseteq\mathcal H^1
\subseteq\mathcal D(A_\pm),
\label{eq:15.2.8.3.7}
\end{align}
and $A_\pm$ coincide with the (maximal)  distributional differential operator $A$ on $\mathcal D(A_\pm)$:
\begin{align}
A_\pm=A=\mathop{\mathrm{Re}} p^r
=\tfrac12\bigl( p^r+( p^r)^*\bigr)\quad \mbox{on } 
\mathcal D(A_\pm),
\label{eq:15.2.8.3.8}
\end{align}
respectively, cf.\ \eqref{eq:13.9.23.2.24}.
In fact $\mathcal D(A_+)$ is exactly the domain of the  maximal
distributional differential operator $A=\tfrac12\bigl( p^r+(
p^r)^*\bigr)$. 
We may call $A=A_+$ the \textit{conjugate operator associated with}
$r$  although  $A_-$ would deserve the same name due 
to the inclusion relations \eqref{eq:15.2.8.3.7}. 
We are going to demonstrate that the
corresponding commutator with $H$ tends to be positive although this
will be in a different sense from that of \cite{Mo}.

\subsection{Commutators with weight inside}\label{sec:13.9.6.23.39}
Using (\ref{eq:11.3.22.6.20}) and \eqref{eq:contract}
we could compute the simple commutator $[H,\mathrm{i}A]$.
However, in the later sections we shall actually use more general commutators 
with a weight $\Theta$ inside:
\begin{align}
[H,\mathrm iA]_\Theta:=\mathrm i(H\Theta A-A\Theta H).
\label{eq:150910}
\end{align}
Hence, in this subsection, we explicitly compute the weighted commutator \eqref{eq:150910}, 
and precisely formulate 
how we should realize it as an operator. 
To our knowledge it is a technical novelty of this
paper. Anyway the quantity \eqref{eq:150910} may be seen as  a technical device facilitating proofs.

Let $\Theta=\Theta(r)$ be a non-negative smooth function only of $r$ with bounded derivatives.
More explicitly, if we denote its derivatives in $r$ by primes such as $\Theta'$, then
\begin{align}
\Theta\ge 0,\quad 
|\Theta^{(k)}|\le C_k,\quad 
k=0,1,2,\dots.
\label{eq:150921}
\end{align}
We first define the weighted commutator \eqref{eq:150910} as a quadratic form on
$C^\infty_{\mathrm c}(M)$, and then extend it onto $\mathcal H^1$ by the following lemma.
Throughout the paper we shall always use the notation $[H,\mathrm iA]_\Theta$
in this extended sense.
\begin{lemma}\label{lem:12.7.3.0.45}
Suppose Conditions~\ref{cond:12.6.2.21.13}--\ref{cond:10.6.1.16.24},
and let $\Theta$ be a non-negative smooth function of $r$ with bounded derivatives
\eqref{eq:150921}.
Then, as quadratic forms on $C^\infty_{\mathrm{c}}(M)$,
\begin{align}
\begin{split}
[H,\mathrm iA]_\Theta
&=
A \Theta' A
+p_i^*\Theta (\nabla^2 r)^{ik}p_k
-\mathop{\mathrm{Im}}\bigl(\Theta'(\mathrm dr)_i(\nabla^2 r)^{ij}p_j\bigr)
-\tfrac14|\mathrm dr|^4\Theta'''
\\&\phantom{{}={}}
+q_\Theta
-2\mathop{\mathrm{Im}}\bigl(q_2\Theta p^r\bigr)
-\tfrac12\mathop{\mathrm{Im}}\bigl(\Theta(\nabla_i\Delta r)\ell^{ij}p_j\bigr)
-\mathop{\mathrm{Re}}\bigl(|\mathrm dr|^2\Theta'H\bigr)
;
\end{split}\label{eq:15091015}
\\
\begin{split}
q_\Theta&=
-(\nabla^rq_1)\Theta
+q_2(\Delta r)\Theta 
+\tfrac18(\nabla^r\tilde\eta)(\Delta r)^2\Theta
\\&\phantom{{}={}}
+\tfrac14(1-\eta)(\nabla^r\Delta r)\Theta'
+|\mathrm dr|^2q_2\Theta'
-\tfrac14(\nabla^r|\mathrm dr|^2)\Theta''
.
\end{split}
\nonumber
\end{align}
In particular by formally absorbing $q_1$ into $q_2$
(by undoing commutation on $q_1$)
and using the Cauchy-Schwarz inequality $[H,\mathrm{i}A]_\Theta$ restricted
to  $C^\infty_{\mathrm{c}}(M)$
extends to a bounded form on $\mathcal H^1$, and whence
$[H,\mathrm{i}A]_\Theta$  can be regarded as a bounded operator $\mathcal H^1\to\mathcal H^{-1}$.
\end{lemma}
\begin{remarks}\label{rem:16020515}
Of course, there are several variations of 
the expression for $[H,\mathrm A]_\Theta$.
\begin{enumerate}
\item
To verify the latter part of the assertion 
the expression \eqref{eq:14.9.4.18.22ffaabbqqee} below
would be more natural and convenient. 
However, in our application in the later sections, we will always use the 
expression \eqref{eq:15091015}.

\item\label{item:1602051525}
In \cite{K4} Kumura obtains the limiting absorption principle 
 with $V\equiv 0$  but  without any regularity on $\Delta r$ (for an exact distance function
$r$). This involves  strong asymptotics hypotheses, in particular
it is  required in  hyperbolic type ends that 
\begin{align}
\Delta r=\beta+O(r^{-\delta});\quad\beta>0,\,\delta>1.
\label{eq:160207}
\end{align}
 However Kumura does not impose bounds on derivatives.
In such a low regularity setting 
it would be more useful to utilize 
an  expression avoiding the derivatives of $\Delta r$, 
rather than \eqref{eq:15091015}. Here is such an alternative:
\begin{align*}
\begin{split}
[H_0,\mathrm iA]_\Theta
&=
A\Theta' A
+p_i^*\Theta (\nabla^2 r)^{ik}p_k
-\mathop{\mathrm{Im}}\bigl(\Theta'(\mathrm dr)_i(\nabla^2 r)^{ij}p_j\bigr)
-\tfrac14|\mathrm dr|^4\Theta'''
\\&\phantom{{}={}}
-\tfrac14(\nabla^r|\mathrm dr|^2)\Theta''
-\tfrac14|\mathrm dr|^2(\Delta r)\Theta''
+\tfrac12\mathop{\mathrm{Im}}\bigl((\Delta r)\Theta'A\bigr)
\\&\phantom{{}={}}
+\mathop{\mathrm{Re}}\bigl((\Delta r)\Theta H_0\bigr)
-\tfrac12\mathop{\mathrm{Re}}\bigl(\Theta p_i^*(\Delta r)g^{ij}p_j\bigr)
-\mathop{\mathrm{Re}}\bigl(|\mathrm dr|^2\Theta'H_0\bigr).
\end{split}
%\label{eq:160205}
\end{align*}
 For example under the condition \eqref{eq:160207}, due to cancellations
\begin{align*}
\begin{split}
[H_0,\mathrm iA]_\Theta
&=
A\Theta' A
+p_i^*\Theta (\nabla^2 r)^{ik}p_k
-\mathop{\mathrm{Im}}\bigl(\Theta'(\mathrm dr)_i(\nabla^2 r)^{ij}p_j\bigr)
-\tfrac14|\mathrm dr|^4\Theta'''
\\&\phantom{{}={}}
-\tfrac14(\nabla^r|\mathrm dr|^2)\Theta''
-\tfrac14O(r^{-\delta})\Theta''
+\tfrac12\mathop{\mathrm{Im}}\bigl(O(r^{-\delta})\Theta'A\bigr)
\\&\phantom{{}={}}
+\mathop{\mathrm{Re}}\bigl(O(r^{-\delta})\Theta H_0\bigr)
-\tfrac12p_i^* O(r^{-\delta})\Theta g^{ij}p_j
-\mathop{\mathrm{Re}}\bigl(|\mathrm dr|^2\Theta'H_0\bigr).
\end{split}
%\label{eq:160205}
\end{align*}
The last expression has the same leading terms as those of
\eqref{eq:15091015}, and hence the proofs of the following sections
should be applicable in particular to the setting of \cite{K4}.  There
would also appear derivatives of $\Delta r$ in these proofs, but after
`undoing of the commutator' they would contribute to a remainder
term only (denoted there by `$Q$'). We shall not give details.
We note that \eqref{eq:160207} reasonably may be
  called `short-range'. On the other hand 
  Example \ref{examples:radi-cond-appl} \eqref{item:13.9.27.20.54}
   includes   `long-range' hyperbolic type ends.

\item\label{item:160208}
More generally than \eqref{eq:160207},
if we assumed short- and long-range splitting conditions for $\Delta r$ as in \cite{IS2}, 
it should be possible to obtain and slightly generalize the results there. 
Nevertheless,  for simplicity of presentation,  we shall not elaborate on details.

\end{enumerate}
\end{remarks}
For the proof of Lemma~\ref{lem:12.7.3.0.45}
we shall use the following commutator computation:
\begin{lemma}\label{lem:12.10.12.8.57}
Let $\tilde{g}\in \Gamma(T^{0,2})$ be symmetric,
and set 
\begin{align*}
\tilde H_0=-\tfrac12\tilde\Delta=\tfrac12p_i^*\tilde{g}^{ij}p_j
.
\end{align*}
Then, as a quadratic form on $C^\infty_{\mathrm{c}}(M)$,
\begin{align*}
\begin{split}
[\tilde{H}_0,\mathrm{i}A]
&=\tfrac12p_i^*\bigl\{
\tilde{g}^i{}_j (\nabla^2r)^{jk}
+(\nabla^2 r)^i{}_j\tilde{g}^{jk}
-(\nabla^r \tilde{g})^{ik} \bigr\}p_k
-\tfrac14 (\tilde \Delta\Delta r).
\end{split}%\label{eq:12.7.2.23.39}
\end{align*}
\end{lemma}
\begin{proof}
Noting the expression \eqref{eq:12.6.27.0.43}
and the general identity holding for any $f\in C^\infty(M)$:
\begin{align}
p_i^*f \tilde g^{ij}p_j=2\mathop{\mathrm{Re}}(f \tilde H_0)+\tfrac12(\tilde\Delta f),
\label{eq:13.9.7.9.32}
\end{align}
we have 
\begin{align}
\begin{split}
[\tilde H_0,\mathrm{i}A]
&
=2\mathop{\mathrm{Im}}\bigl((p^r)^*\tilde H_0\bigr)
+\mathop{\mathrm{Re}}\bigl((\Delta r)\tilde H_0\bigr)
\\&
=2\mathop{\mathrm{Re}}(\tilde H_0\nabla^r)
+\tfrac12p_i^*(\Delta r)\tilde g^{ij}p_j
-\tfrac14(\tilde\Delta \Delta r).
\end{split}
\label{eq:14.9.4.18.22ffaabbzz}
\end{align}
Let us compute the first term on the right-hand side of \eqref{eq:14.9.4.18.22ffaabbzz} 
in the form of expectation.
This actually prevents unnecessary complication otherwise coming from covariant derivatives on higher-order tensors.
For any state $\psi\in C^\infty_{\mathrm c}(M)$
\begin{align}
\begin{split}
&\bigl\langle2\mathop{\mathrm{Re}}(\tilde H_0\nabla^r)\bigr\rangle_\psi
\\&=
\mathop{\mathrm{Re}}\bigl\langle \nabla^*\tilde g^{\bullet i}(\nabla\psi)_i,  \nabla^r\psi\bigr\rangle
\\&
=
\mathop{\mathrm{Re}}\bigl\langle \tilde g^{ji}(\nabla\psi)_i, (\nabla^2 r)_j{}^k(\nabla\psi)_k\bigr\rangle
+\mathop{\mathrm{Re}}\bigl\langle \tilde g^{ji}(\nabla\psi)_i,  (\nabla r)^k(\nabla^2\psi)_{jk}\bigr\rangle
\\&
=
\mathop{\mathrm{Re}}\bigl\langle p_i^* \tilde g^{ij}(\nabla^2 r)_j{}^kp_k\bigr\rangle_\psi
\\&\phantom{={}}
+\tfrac12\Bigl[
\bigl\langle (\nabla\psi)_i,  \tilde g^{ij}(\nabla r)^k(\nabla^2\psi)_{kj}\bigr\rangle
+\bigl\langle   (\nabla^2\psi)_{kj},\tilde g^{ji}(\nabla r)^k(\nabla\psi)_i\bigr\rangle
\Bigr]\\&
=
\tfrac12\bigl\langle p_i^* \bigl\{\tilde g^{ij}(\nabla^2 r)_j{}^k
+(\nabla^2 r)^i{}_j\tilde g^{jk}-(\nabla^r\tilde g)^{ik}\bigr\}p_k\bigr\rangle_\psi
-\tfrac12\bigl\langle   p_j^*(\Delta r)\tilde g^{ji}p_i\bigr\rangle_\psi
.
\end{split}
\label{eq:14.9.4.18.22ffaabbzzaa}
\end{align}
Hence by \eqref{eq:14.9.4.18.22ffaabbzz} and \eqref{eq:14.9.4.18.22ffaabbzzaa}
the assertion follows.
\end{proof}

\begin{proof}[Proof of Lemma~\ref{lem:12.7.3.0.45}]
By \eqref{eq:15091010} and \eqref{eq:13.9.23.5.54}
we can compute
\begin{align}
\begin{split}
[H,\mathrm iA]_\Theta
&=
\mathop{\mathrm{Im}}\bigl(A\Theta A\tilde\eta A\bigr)
+\mathop{\mathrm{Im}}(A\Theta L )
+2\mathop{\mathrm{Im}}(A\Theta q)
-\tfrac12\mathop{\mathrm{Im}}\bigl((\nabla^r\tilde\eta )(\Delta r)\Theta A\bigr)
\\&=
\tfrac12 A \eta\Theta' A
-\tfrac12 A (\nabla^r\tilde\eta)\Theta A
+\tfrac12\bigl[p_i^*\Theta\ell^{ij}p_j,\mathrm i A\bigr]
+\mathop{\mathrm{Re}}\bigl(A (1-\eta)\Theta'p^r\bigr)
\\&\phantom{{}={}}
-|\mathrm dr|^2q_1\Theta'
-(\nabla^rq_1)\Theta
-2\mathop{\mathrm{Im}}(q_2\Theta A)
-\tfrac12\mathop{\mathrm{Im}}\bigl((\nabla^r\tilde\eta )(\Delta r)\Theta A\bigr)
.
\end{split}\label{eq:14.9.4.18.22ffaabbqqe}
\end{align} 
To compute the third term on the right-hand side of \eqref{eq:14.9.4.18.22ffaabbqqe}
we apply Lemma~\ref{lem:12.10.12.8.57} with $\tilde g=\Theta\ell$.
We also use \eqref{eq:13.9.23.5.54}, \eqref{eq:15120921}
and \eqref{eq:12.6.27.0.43},
and then we can combine the third and eighth terms of \eqref{eq:14.9.4.18.22ffaabbqqe} as
\begin{align}
\begin{split}
&\tfrac12\bigl[p_i^*\Theta\ell^{ij}p_j,\mathrm i A\bigr]
-\tfrac12\mathop{\mathrm{Im}}\bigl((\nabla^r\tilde\eta )(\Delta r)\Theta A\bigr)
\\&
=
\tfrac12p_i^*\Bigl(
\Theta\ell^i{}_j(\nabla^2 r)^{jk}
+\Theta(\nabla^2 r)^i{}_j\ell^{jk}
-(\nabla^r \Theta\ell)^{ik}\Bigr)p_k
\\&\phantom{{}={}}
+\tfrac14\bigl(p_i^*\Theta\ell^{ij}p_j\Delta r\bigr)
-\tfrac12\mathop{\mathrm{Im}}\bigl((\nabla^r\tilde\eta )(\Delta r)\Theta A\bigr)
\\&=
\tfrac12p_i^*\Bigl(
2\Theta (\nabla^2 r)^{ik}
-|\mathrm dr|^2\Theta'\ell^{ik}
+(\nabla^r \tilde\eta )\Theta(\mathrm dr\otimes\mathrm dr)^{ik}
\Bigr)p_k
\\&\phantom{{}={}}
-\tfrac12\mathop{\mathrm{Im}}\bigl(\Theta (\nabla_j\Delta r)\ell^{ij}p_i\bigr)
-\tfrac12\mathop{\mathrm{Im}}\bigl((\nabla^r\tilde\eta )(\Delta r)\Theta p^r\bigr)
+\tfrac14(\nabla^r\tilde\eta )(\Delta r)^2\Theta
\\&=
p_i^*\Theta (\nabla^2 r)^{ik}p_k
-\tfrac12p_i^*|\mathrm dr|^2\Theta'\ell^{ik}p_k
+\tfrac12A(\nabla^r \tilde\eta )\Theta A
\\&\phantom{{}={}}
-\tfrac12\mathop{\mathrm{Im}}\bigl(\Theta (\nabla_j\Delta r)\ell^{ij}p_i\bigr)
+\tfrac18(\nabla^r\tilde\eta)(\Delta r)^2\Theta
.
\end{split}\label{eq:14.9.17.12.15ffaabccbqqe}
\end{align}
By \eqref{eq:12.6.27.0.43}
we write the fourth and seventh terms of \eqref{eq:14.9.4.18.22ffaabbqqe} as 
\begin{align}
\begin{split}
&
\mathop{\mathrm{Re}}\bigl(A (1-\eta)\Theta'p^r\bigr)
-2\mathop{\mathrm{Im}}(q_2\Theta A)
\\
&=
A (1-\eta)\Theta'A
-\tfrac14(\nabla^r \eta)(\Delta r)\Theta'
+\tfrac14(1-\eta)(\nabla^r \Delta r)\Theta'
\\&\phantom{{}={}}
+\tfrac14(1-\eta)|\mathrm dr|^2(\Delta r)\Theta''
+q_2 (\Delta r)\Theta
-\mathop{\mathrm{Im}}\bigl(2q_2\Theta p^r\bigr)
.
\end{split}
\label{eq:1509101044}
\end{align}
By \eqref{eq:14.9.4.18.22ffaabbqqe}, \eqref{eq:14.9.17.12.15ffaabccbqqe} and \eqref{eq:1509101044}
it follows that 
\begin{align}
\begin{split}
[H,\mathrm iA]_\Theta
&=
\tfrac12A (2-\eta)\Theta'A
+p_i^*\Theta (\nabla^2 r)^{ik}p_k
-\tfrac12p_i^*|\mathrm dr|^2\Theta'\ell^{ik}p_k
-(\nabla^rq_1)\Theta
\\&\phantom{{}={}}
+q_2 (\Delta r)\Theta
-\tfrac12\mathop{\mathrm{Im}}\bigl(\Theta (\nabla_i\Delta r)\ell^{ij}p_i\bigr)
+\tfrac18(\nabla^r\tilde\eta)(\Delta r)^2\Theta
-|\mathrm dr|^2q_1\Theta'
\\&\phantom{{}={}}
-\tfrac14(\nabla^r \eta)(\Delta r)\Theta'
+\tfrac14(1-\eta)(\nabla^r\Delta r)\Theta'
+\tfrac14(1-\eta)|\mathrm dr|^2(\Delta r)\Theta''
\\&\phantom{{}={}}
-\mathop{\mathrm{Im}}\bigl(2q_2\Theta p^r\bigr)
.
\end{split}\label{eq:14.9.4.18.22ffaabbqqee}
\end{align} 
We combine a part of the first term, 
the third and eighth terms of \eqref{eq:14.9.4.18.22ffaabbqqee}
in the following manner.
We make use of the expressions \eqref{eq:12.6.27.0.43}, \eqref{eq:15091010},
\eqref{eq:13.9.23.5.54} and \eqref{eq:15120921}, and then obtain
\begin{align}
\begin{split}
&
-\tfrac12 A \eta\Theta' A
-\tfrac12p_i^*|\mathrm dr|^2\Theta'\ell^{ik}p_k
-|\mathrm dr|^2q_1\Theta'
\\&
=
-\tfrac12 \mathop{\mathrm{Im}}
\Bigl[
\bigl(\nabla^r|\mathrm dr|^2\Theta'\bigr)\tilde\eta A
+\bigl(\nabla_i|\mathrm dr|^2\Theta'\bigr)\ell^{ik}p_k
\Bigr]
\\&\phantom{{}={}}
-\tfrac12 \mathop{\mathrm{Re}}
\Bigl[
|\mathrm dr|^2\Theta' A\tilde\eta A
+|\mathrm dr|^2\Theta'p_i^*\ell^{ij}p_j
\Bigr]
-|\mathrm dr|^2q_1\Theta'
\\&
=
-\tfrac12 \mathop{\mathrm{Im}}
\bigl[\bigl(\nabla_i|\mathrm dr|^2\Theta'\bigr)g^{ij}p_j\bigr]
+\tfrac14\bigl(\nabla^r|\mathrm dr|^2\Theta'\bigr)\tilde\eta  (\Delta r)
\\&\phantom{{}={}}
-\mathop{\mathrm{Re}}\bigl(|\mathrm dr|^2\Theta'H\bigr)
+\tfrac14|\mathrm dr|^2(\nabla^r\tilde\eta )(\Delta r)\Theta'
+|\mathrm dr|^2q_2\Theta'
\\&
=
-\mathop{\mathrm{Im}}\bigl[\Theta'(\mathrm dr)_i(\nabla^2r)^{ij}p_j\bigr]
-\tfrac14(\nabla^r|\mathrm dr|^2)\Theta''
-\tfrac14|\mathrm dr|^4\Theta'''
\\&\phantom{{}={}}
-\tfrac14(1-\eta) |\mathrm dr|^2(\Delta r)\Theta''
-\mathop{\mathrm{Re}}\bigl(|\mathrm dr|^2\Theta'H\bigr)
+\tfrac14(\nabla^r\eta)(\Delta r)\Theta'
+|\mathrm dr|^2q_2\Theta'.
\end{split}\label{eq:15091012}
\end{align}
If we substitute \eqref{eq:15091012} into \eqref{eq:14.9.4.18.22ffaabbqqee},
then the expression \eqref{eq:15091015} follows.

It remains to show the boundedness of $[H,\mathrm iA]_\Theta$
as an operator $\mathcal H^1\to\mathcal H^{-1}$,
but it is obvious by \eqref{eq:14.9.4.18.22ffaabbqqee} or \eqref{eq:15091015} and 
Conditions~\ref{cond:12.6.2.21.13}--\ref{cond:10.6.1.16.24}. 
\end{proof}

\subsection{Doing and undoing commutators}

In the previous subsection 
we defined the weighted commutator $[H,\mathrm{i}A]_\Theta$
as a quadratic form on $\mathcal H^1$
by the extension from $C^\infty_{\mathrm{c}}(M)$.
On the other hand, throughout the paper, we shall use the notation 
\begin{align}
\mathop{\mathrm{Im}}(A\Theta H)=\tfrac1{2\mathrm i}(A\Theta H-H\Theta A)
\label{eq:150911}
\end{align}
as a quadratic form defined on $\mathcal D(H)$, i.e.\ for $\psi\in \mathcal D(H)$
\begin{align*}
\bigl\langle \mathop{\mathrm{Im}}(A\Theta H)\bigr\rangle_\psi
=\tfrac1{2\mathrm i}\bigl(\langle A\psi,\Theta H\psi\rangle 
-\langle H\psi,\Theta A\psi\rangle\bigr).
\end{align*}
Obviously the quadratic forms $[H,\mathrm iA]_\Theta$ and $2\mathop{\mathrm{Im}}(A\Theta H)$
coincide on $C^\infty_{\mathrm c}(M)$,
but they do not in general on $\mathcal D(H)$.
This is due to the third order terms in \eqref{eq:150911}.
Although the third order terms themselves finally cancel out after integrations by parts,
the remaining ``boundary contribution'' is  not
negligible, see Remark \ref {rem:boundary}.
Nonetheless, fortunately, these boundary contributions have
sign, and Lemma~\ref{lem:undoing-commutators} below 
allows us to `do' and `undo' the commutator to some extent.
\begin{lemma}\label{lem:undoing-commutators}
Suppose Conditions~\ref{cond:12.6.2.21.13}--\ref{cond:10.6.1.16.24},
and let $\Theta$ be a non-negative smooth function of $r$ with bounded derivatives
\eqref{eq:150921}.
Then, as quadratic form on $\mathcal D(H)$, 
\begin{align}
[H,\mathrm iA]_\Theta\le 2\mathop{\mathrm{Im}}(A\Theta H).
\label{eq:150911440}
\end{align}
\end{lemma}

In this subsection we prove Lemma~\ref{lem:undoing-commutators}.
Let us denote the Dirichlet self-adjoint realization of the operator
\begin{align*}
H_\Theta=\tfrac12p_i^*\Theta g^{ij}p_j+\Theta V\quad \text{on }C^\infty_{\mathrm c}(M)
\end{align*}
by the same symbol $H_\Theta$.
We denote its operator and form domains by $\mathcal D(H_\Theta)$
and $\mathcal H^1_\Theta$, respectively,
and then obviously we have
\begin{align}
\mathcal H^1\subseteq \mathcal H^1_\Theta,\quad
\mathcal D(H)\subseteq \mathcal D(H_\Theta),
\label{eq:150911427}
\end{align} 
cf.\ Subsection~\ref{subsec:result}.
It follows that, as quadratic forms on $\mathcal D(H)$, 
\begin{align}
2\mathop{\mathrm{Im}}(A\Theta H)
=
2\mathop{\mathrm{Im}}(AH_\Theta)
+\mathop{\mathrm{Re}}(A \Theta'p^r)
\label{eq:150911414}
\end{align}
and that, as quadratic forms on $\mathcal H^1$ (as extensions from $C^\infty_{\mathrm c}(M)$),
\begin{align}
[H,\mathrm iA]_\Theta
=
[H_\Theta,\mathrm iA]
+\mathop{\mathrm{Re}}(A \Theta'p^r).
\label{eq:150911415}
\end{align}
The second operators on the right-hand sides of \eqref{eq:150911414} and \eqref{eq:150911415}
clearly coincide on $\mathcal D(H)$,
and hence the proof of Lemma~\ref{lem:undoing-commutators}
reduces to that of
\begin{align}
[H_\Theta,\mathrm iA]\le 2\mathop{\mathrm{Im}}(AH_\Theta),
\label{eq:150911422}
\end{align}
as quadratic forms on $\mathcal D(H)$.

In order to prove \eqref{eq:150911422}
we first prove regularity properties of the flow \eqref{eq:12.6.7.1.10c}.

\begin{lemma}\label{lem:12.6.4.13.14} 
Suppose Conditions~\ref{cond:12.6.2.21.13}--\ref{cond:10.6.1.16.24} .
Then for any $t\ge 0$ one has the natural bounded extension/restriction
$T(\pm t)\colon \mathcal H^{\mp 1}\to \mathcal H^{\mp 1}$,
and   
\begin{align}
    \sup_{t\in [0,1]}
 \|T(\pm t)\|_{{\mathcal B}(\mathcal H^{\mp 1})}<\infty,
    \label{eq:12.6.7.2.34}
  \end{align} 
respectively.
  Moreover, $T(\pm t)\in {\mathcal B}(\mathcal H^{\mp 1})$
are strongly continuous in $t\ge 0$, respectively.
\end{lemma}
\begin{proof}
It suffices to prove the assertions for $T(-t)$, $t\ge 0$, since 
those for $T(t)$, $t\ge 0$, follow by taking the adjoint, cf.\ \cite[Theorem 10.6.5]{HP}. 
For any $\psi \in C^\infty_{\mathrm{c}}(M)$ 
we have by \eqref{eq:12.6.7.1.10c} and standard regularity properties for
flows that
  \begin{align}
    \begin{split}
    p_i(T(-t)\psi)(x)
    &=[\partial_iy^\alpha(-t,x)](T(-t)p_\alpha\psi)(x)
    \\&\phantom{{}={}}+\left(\int_0^{-t}\tfrac1{2\mathrm i}
      [\partial_iy^\alpha(s,x)]
(\partial_\alpha\Delta r)(y(s,x))\,\mathrm{d}s \right)
    (T(-t)\psi)(x)
  \end{split}\label{eq:12.7.25.2.39}
  \end{align}
  if $(-t,x)\in\mathcal M$, and $p_i(T(-t)\psi)(x)=0$ otherwise. 
We note that by definition for $(-t,x)\not\in\mathcal M$ the factor  $
(T(-t)\psi)(x)=0$. Repeated differentiation leads to the conclusion
that $T(-t)\psi\in C^\infty_{\mathrm c}(M)\subseteq \mathcal H^1$. 

It follows readily from \eqref{eq:12.6.7.1.10c} and
\eqref{eq:12.7.25.2.39} that the $\vH^1$-valued function $T(-t)\psi$
for $\psi \in C^\infty_{\mathrm{c}}(M)$
is continuous in $t\geq 0$. Given the boundedness
\eqref{eq:12.6.7.2.34} we would then obtain the strong continuity of
$T(- t)\in {\mathcal B}(\mathcal H^{ 1})$ by a density argument. 
Hence it remains to show \eqref{eq:12.6.7.2.34} for $T(-t)$.
We shall prove 
\begin{align*}
\inp{H_0+1}_{T(-t)\psi}\le C_1
\end{align*}
independently of $t\in [0,1]$ 
and $\psi \in C^\infty_{\mathrm{c}}(M)$ with $\|\psi\|_{\mathcal H^{ 1}}=1$,
and for that it suffices to bound
\begin{align*}
f(t):=\inp{H+C_2}_{T(-t)\psi}\ge \inp{H_0+1}_{T(-t)\psi};\quad
C_2=1+\|V\|_{L^\infty}
,
\end{align*}
above.
By Lemma~\ref{lem:12.7.3.0.45} we indeed have
$C_3:=\|[H,\i A]\|_{\vB(\vH^1,\vH^{-1})}<\infty$,
and then
\begin{align*}
 f'(t)=-\langle [H,\i A]\rangle_{T(-t)\psi}
\leq C_3\|T(-t)\psi\|^2_{\mathcal H^{ 1}} \le C_4 f(t).
\end{align*}
This estimate leads to $f(t)\leq f(0)\e^{tC_4}$, and we are done. 
\end{proof}

\begin{lemma}\label{lem:12.6.3.18.30}
Under Conditions~\ref{cond:12.6.2.21.13}--\ref{cond:10.6.1.16.24}
  there exists $C>0$ such that for any $t\in [0,1]$
  \begin{align*}
    \|H_\Theta-T(t)H_\Theta T(-t)\|_{{\mathcal
        B}(\mathcal H^{1},\mathcal H^{-1})}&\le Ct.%\label{eq:13.7.28.0.55}
  \end{align*}
\end{lemma}
\begin{proof}
By  \eqref{eq:15.2.8.3.7}, as quadratic forms on $C^\infty_{\mathrm{c}}(M)$, 
  \begin{align*}
   H_\Theta- T(t)H_\Theta T(-t)
    =\int_0^t
    T(s)[H_\Theta,\mathrm iA]T(-s)\,\mathrm{d}s.
  \end{align*}
  By \eqref{eq:150911427}, Lemma~\ref{lem:12.6.4.13.14}  and the denseness  of 
  $C^\infty_{\mathrm{c}}(M)\subseteq \mathcal H^{1}$ the
  assertion follows.
\end{proof}

\begin{lemma}\label{lem:12.6.3.18.30b}
 Under Conditions~\ref{cond:12.6.2.21.13}--\ref{cond:10.6.1.16.24}
the commutator $[H_\Theta,\mathrm iA]$ has the expression
  \begin{align}
    [H_\Theta,\mathrm iA] &=\slim_{t\to 0^+}
    t^{-1}(H_\Theta-T(t)H_\Theta T(-t))
    \quad\mbox{in }\mathcal B(\mathcal H^{1},\mathcal H^{-1}).
    \label{eq:12.6.4.3.29}
    \end{align}
\end{lemma}
\begin{proof}
We write 
for any $\psi\in C^\infty_{\mathrm{c}}(M)$
  \begin{align*}
   & t^{-1}\bigl(H_\Theta-T(t)H_\Theta T(-t)\bigr)\psi
-[H_\Theta,\mathrm iA]\psi
\\&    =t^{-1}\int_0^t \bigl\{T(s)[H_\Theta,\mathrm iA]T(-s)-[H_\Theta,\mathrm iA]\bigr\}
     \psi\,\mathrm{d}s.
  \end{align*}
  Then we obtain 
  (\ref{eq:12.6.4.3.29}) on $C^\infty_{\mathrm{c}}(M)$
  due to the strong continuity of $T(\pm t)$ stated in Lemma~\ref{lem:12.6.4.13.14}.
  Then in turn by Lemma~\ref{lem:12.6.3.18.30} and the density
  argument, the strong limit to the right of  (\ref{eq:12.6.4.3.29}) exists in
  ${\mathcal B}(\mathcal H^{1},\mathcal H^{-1})$ and the equality holds.  
\end{proof}
\begin{proof}[Proof of Lemma~\ref{lem:undoing-commutators}]
It suffices to show \eqref{eq:150911422} on $\mathcal D(H)$.
Due to the non-negativity of $\Theta$,
we have the inequality, as quadratic forms on $\mathcal H^1$,
\begin{align*}
&H_\Theta-T(t)H_\Theta T(-t)
\\&
=
H_\Theta(1-T(-t))
+(1-T(-t))^*H_\Theta
-(1-T(-t))^*H_\Theta(1-T(-t))\\
&\le 
H_\Theta(1-T(-t))
+(1-T(-t))^*H_\Theta
-(1-T(-t))^*\Theta V(1-T(-t)).
\end{align*}
We evaluate this inequality in the state
$\psi\in \mathcal D(H)\subseteq\mathcal H_1$,
divide it by $t>0$ and take the limit $t\to 0^+$ using 
Lemma~\ref{lem:12.6.3.18.30b} and \eqref{eq:15.2.8.3.7}.
Then we obtain 
\begin{align}
\begin{split}
\langle [H_\Theta,\mathrm iA]\rangle_{\psi}
&=\slim_{t\to 0^+}
    t^{-1}\langle H_\Theta-T(t)H_\Theta T(-t)\rangle_\psi\\
&\le 
\slim_{t\to 0^+}
    t^{-1}
\Bigl\{\bigl\langle H_\Theta\psi, (1-T(-t))\psi\bigr\rangle
+\bigl\langle (1-T(-t))\psi, H_\Theta\psi\bigr\rangle
\\&\phantom{\le\slim_{t\to +0}t^{-1}\Bigl\{}
-\bigl\langle (1-T(-t))\psi, \Theta V(1-T(-t))\psi\bigr\rangle\Bigr\}
\\   
    &= \langle H_\Theta\psi,
    \mathrm{i}A\psi\rangle + \langle
    \mathrm{i}A\psi, H_\Theta\psi\rangle.
\end{split}\label{eq:12.6.4.3.31}
  \end{align}
Hence we are done.
\end{proof}
\begin{remark}\label{rem:boundary}
  We can also prove that, if the gradient vector field $\omega$ is
  both forward and backward complete, then the equality holds in
  \eqref{eq:12.6.4.3.31}, and hence also in \eqref{eq:150911440}.  In
  fact we can prove this by using, instead of \eqref{eq:12.6.4.3.29},
  an alternative expression
  \begin{align*}
    [H_\Theta,\mathrm iA] &=\slim_{t\to 0}
    t^{-1}(H_\Theta T(t)-T(t)H_\Theta)
    \quad\mbox{in }\mathcal B(\mathcal H^{1},\mathcal H^{-1})
    \end{align*}
holding true in this case. 
As we can see in the proof 
the gap between  the left- and right-hand side  quantities of \eqref{eq:12.6.4.3.31}
is given by 
\begin{align}
\slim_{t\to 0^+}t^{-1}\bigl\|\sqrt \Theta p(1-T(-t))\psi\bigr\|^2.
\label{eq:15.2.8.15.38}
\end{align}
Assuming only   forward completeness on $\omega$ it does not vanish
in general.  Actually the gap \eqref{eq:15.2.8.15.38} corresponds
formally to a `boundary contribution' appearing from integrations by
parts.  In  cases where somehow we can  realize a smooth
boundary $\partial M$ of $M$, such as the open half-space of the
Euclidean space, indeed we can explicitly compute the gap
\eqref{eq:15.2.8.15.38}, and it is the square of a weighted $L^2$-norm
of the normal derivative of $\psi$ on $\partial M$. 
  See \cite[Proposition 6.2]{BGS} for a  concrete very similar 
  computation. Of course, the
weight vanishes if $\omega$ is both forward and backward complete,
i.e.\ parallel to this boundary.
\end{remark}

\section{Rellich's theorem}\label{sec:absence} 
Our proof of Theorem~\ref{thm:13.6.20.0.10}
shares features of the standard scheme \cite{FHH2O,FH,IS2}
used for showing  absence of `$L^2$-eigenvalues'. However there are
notable differences. 
Our main novelty  is the use of the conjugate operator $A$
associated with $r$ rather than the one  associated with  $r^2$, 
cf.  Subsections~\ref{subsec:result} and \ref{subsec:15.2.5.21.22}.
For such $A$
the formal commutator $[H,\mathrm iA]$ has only 
a weaker and partial positivity (in the spherical direction),
but refined arguments finally provide a stronger result.
For the Euclidean space our result overlaps with  \cite{L1,L2} and in
particular with \cite[Section 30.2]{Ho2}.

Basically the proof consists of two steps,
a priori super-exponential decay estimates
and the absence of super-exponentially decaying eigenfunctions.
Obviously,  Theorem~\ref{thm:13.6.20.0.10}
follows immediately as a combination of the following propositions.
Throughout the section 
we suppose Conditions~\ref{cond:12.6.2.21.13}--\ref{cond:10.6.1.16.24}.

\begin{proposition}\label{prop:absence-eigenvalues-1b}
Let $\lambda>\lambda_0$.
If a function $\phi\in   \mathcal H_{\mathrm{loc}}$ satisfies for
some $m_0\ge 0$:
\begin{enumerate}
\item
$(H-\lambda)\phi=0$ in the distributional sense, 
\item
$\bar\chi_{m_0}\phi\in B^*_0$ and 
$\chi_{m_0,n}\phi\in\mathcal H^1$ for any $n>m_0$,
\end{enumerate}
then $\bar\chi_{m_0}\mathrm e^{\alpha r}\phi\in B_0^*$ for any $\alpha\ge 0$.
\end{proposition}

\begin{proposition}\label{prop:absence-eigenvalues-1bb}
Let $\lambda>\lambda_0$.
If a function $\phi\in \mathcal H_{\mathrm{loc}}$ satisfies for
some $m_0\ge 0$:
\begin{enumerate}
\item
$(H-\lambda)\phi=0$ in the distributional sense, 
\item
$\bar\chi_{m_0}\mathrm e^{\alpha r}\phi\in B_0^*$ for any $\alpha\ge 0$
and 
$\chi_{m_0,n}\phi\in\mathcal H^1$ for any $n>m_0$,

\end{enumerate}
then $\phi(x)=0$ in $M$.
\end{proposition}

We prove 
Propositions~\ref{prop:absence-eigenvalues-1b} and \ref{prop:absence-eigenvalues-1bb}
in Subsections~\ref{subsec:150909511} and \ref{subsec:150909513},
respectively.
The proofs are quite similar to each other, and 
both are dependent on commutator estimates with a particular form of weight inside.
Let us introduce the regularized weights
\begin{align}
\Theta= \Theta_{m,n,\nu}^{\alpha,\beta,\delta}
=\chi_{m,n}\mathrm e^{\theta};\quad 
n>m\ge 0,
\label{eq:15.2.15.5.8bb}
\end{align}
with exponents 
\begin{align*}
\theta=\theta_\nu^{\alpha,\beta,\delta}
=2\alpha r+2\beta\int_0^r(1+s/R_\nu)^{-1-\delta}\,\mathrm ds;\quad
\alpha,\beta\ge 0,\ \delta>0,\ \nu\ge 0.
\end{align*}
Denote their derivatives in $r$ by primes, e.g., 
if we set for notational simplicity
\begin{align*}
\theta_0=1+r/R_\nu,
\end{align*} 
then
\begin{align*}
\begin{split}
\theta'=2\alpha+2\beta\theta_0^{-1-\delta},\quad
\theta''=-2(1+\delta)\beta R_\nu^{-1}\theta_0^{-2-\delta},\quad \ldots.
\end{split}%\label{eq:12.5.1.19.56}
\end{align*}
In particular, since $R_\nu^{-1}\theta_0^{-1}\leq r^{-1}$, we have
  \begin{align*}
    |\theta^{(k)}|\leq C_{\delta,k} \beta r^{1-k}\theta_0^{-1-\delta};\quad
k=2,3,\dots.
  \end{align*}

\subsection{A priori super-exponential decay estimates}\label{subsec:150909511}

In this subsection we prove Proposition~\ref{prop:absence-eigenvalues-1b}.
The following commutator estimate is a key:
\begin{lemma}\label{lem:14.10.4.1.17ffaabb}
Let $\lambda>\lambda_0$,
and fix any $\alpha_0\ge 0$ and $\delta\in(0,\min\{1,\rho',\tau/2\})$ 
in the definition \eqref{eq:15.2.15.5.8bb} of $\Theta$.
Then there exist $\beta,c,C>0$ and $n_0\ge 0$ such that uniformly in 
$\alpha\in[0,\alpha_0]$, $n>m\ge n_0$ and $\nu\ge n_0$, 
as quadratic forms on $\mathcal D(H)$,
\begin{align}
\begin{split}
\mathop{\mathrm{Im}}
\bigl(A\Theta(H-\lambda)\bigr)
&\ge 
cr^{-1}\theta_0^{-\delta}\Theta
-C\bigl(\chi_{m-1,m+1}^2+\chi_{n-1,n+1}^2\big)r^{-1}\mathrm e^\theta
\\&\phantom{={}}
+\mathop{\mathrm{Re}}\bigl(\gamma(H-\lambda)\bigr),
\end{split}
\label{14.9.26.9.53ffaabb}
\end{align}
where $\gamma=\gamma_{m,n,\nu}$ is a certain function
satisfying $\mathop{\mathrm{supp}}\gamma\subseteq\mathop{\mathrm{supp}}\chi_{m,n}$ 
and $|\gamma|\le C_{m,n}\mathrm e^\theta$.
\end{lemma}
\begin{remark}\label{rem:190802}
The statement is read in the logical order,
and all the arguments except for those concerning $\gamma$ 
are independent of $\alpha,n,m$ and $\nu$.
We shall later take the limits $n\to\infty$ and $\nu\to\infty$,
but contributions from the left-hand side and the last term
vanish beforehand due to the appearance of the factor  $H-\lambda$.
\end{remark}
\begin{proof}
We are going to prove the lemma 
by computing and bounding the quadratic form on the left-hand side of \eqref{14.9.26.9.53ffaabb}.
Fix $\lambda>\lambda_0$ and $\delta\in (0,\min\{1,\rho',\tau/2\})$ as in the assertion.
To avoid confusion for the moment all the estimates below
are uniform in the parameters $\alpha\ge 0$, $\beta\in [0,1]$, $n>m\ge 0$ and $\nu\ge 0$,
so that constants $c_1,C_1,\dots,C_8>0$ are independent of them.
Then in the last step we shall restrict ranges of these parameters
to obtain the assertion.

Recall the notation \eqref{eq:15091112}.
Then by Lemmas~\ref{lem:undoing-commutators}, \ref{lem:12.7.3.0.45}, 
\eqref{eq:12.6.27.0.43}, the Cauchy--Schwarz inequality
and \eqref{eq:14.12.27.3.22}
we can estimate 
\begin{align}
\begin{split}
&\mathop{\mathrm{Im}}\bigl(A\Theta(H-\lambda)\bigr)
\\&
\ge 
\tfrac12A\theta'\Theta A
+\tfrac12p_i^*\Theta h^{ik}p_k
-\tfrac12\mathop{\mathrm{Im}}\bigl(\theta'\Theta(\mathrm dr)_ih^{ij}p_j\bigr)
\\&\phantom{{}={}}
-\tfrac18|\mathrm dr|^4\theta'^3\Theta
-\tfrac38|\mathrm dr|^4\theta'\theta''\Theta
-\tfrac12\mathop{\mathrm{Re}}\bigl(|\mathrm dr|^2\Theta'(H-\lambda)\bigr)
-C_1Q
\\&
\ge 
\tfrac12c_1A\tilde\eta  r^{-1}\theta_0^{-\delta}\Theta A
+\tfrac12c_1p_i^*r^{-1}\theta_0^{-\delta}\Theta \ell^{ij}p_j
\\&\phantom{{}={}}
+\tfrac12A\bigl(\theta'-c_1\tilde\eta  r^{-1}\theta_0^{-\delta}\bigr)\Theta A
+\tfrac14p_i^*\Theta \bigl(h^{ij}-2c_1r^{-1}\theta_0^{-\delta}\ell^{ij}\bigr)p_j
\\&\phantom{{}={}}
-\tfrac18|\mathrm dr|^4\theta'^3\Theta
-\tfrac38|\mathrm dr|^4\theta'\theta''\Theta
-\tfrac12\mathop{\mathrm{Re}}\bigl(|\mathrm dr|^2\Theta'(H-\lambda)\bigr)
-C_2Q
,
\end{split}\label{eq:14.9.4.18.22ffaabbqq}
\end{align} 
where $c_1>0$ is a small constant such that the 
fourth term on the right-hand side of \eqref{eq:14.9.4.18.22ffaabbqq} is non-negative, 
and we have introduced for simplicity
\begin{align}
\begin{split}
Q&=
\Bigl((1+\alpha^2)r^{-1-\min\{1,\rho',\tau/2\}}\chi_{m,n}
+(1+\alpha^2)|\chi_{m,n}'|+(1+\alpha)|\chi_{m,n}''|
\\&\phantom{{}={}\Bigl(}
+|\chi_{m,n}'''|\Bigr)\mathrm e^\theta
+p_i^*\Bigl(r^{-1-\min\{1,\rho',\tau/2\}}\chi_{m,n}+|\chi_{m,n}'|\Bigr)\mathrm e^\theta g^{ij}p_j.
\end{split}\label{eq:15091119}
\end{align}
Let us further compute and estimate the terms on the right-hand side
of \eqref{eq:14.9.4.18.22ffaabbqq}.
Using the expressions \eqref{eq:12.6.27.0.43} and \eqref{eq:15091010}
we estimate the first and second terms of \eqref{eq:14.9.4.18.22ffaabbqq} by
\begin{align}
\begin{split}
&\tfrac12A\tilde\eta  r^{-1}\theta_0^{-\delta}\Theta A
+\tfrac12p_i^*r^{-1}\theta_0^{-\delta}\Theta \ell^{ij}p_j
\\&
\ge 
\tfrac12\mathop{\mathrm{Im}}
\bigl(\eta r^{-1}\theta_0^{-\delta}\theta'\Theta A\bigr)
+
\tfrac12\mathop{\mathrm{Re}}
\Bigl[
r^{-1}\theta_0^{-\delta}\Theta A\tilde\eta  A
+r^{-1}\theta_0^{-\delta}\Theta L \Bigr]
-C_3Q
\\&
\ge 
(\lambda-q_1) r^{-1}\theta_0^{-\delta}\Theta 
+\tfrac14\eta|\mathrm dr|^2r^{-1}\theta_0^{-\delta}\theta'^2\Theta
+\mathop{\mathrm{Re}}
\Bigl[r^{-1}\theta_0^{-\delta}\Theta (H-\lambda)\Bigr]
-C_4Q
.
\end{split}\label{eq:14.9.17.12.15ffaabbqq}
\end{align}
We combine the third, fifth and sixth terms of \eqref{eq:14.9.4.18.22ffaabbqq} as 
\begin{align}
\begin{split}
&
\tfrac12 A \bigl(\theta'-c_1 \tilde\eta r^{-1}\theta_0^{-\delta}\bigr)\Theta A
-\tfrac1{8}|\mathrm dr|^4\theta'^3\Theta
-\tfrac38|\mathrm dr|^4\theta'\theta''\Theta
\\&
\ge 
\tfrac12
\bigl(A+\tfrac{\mathrm i}2|\mathrm dr|^2\theta'\bigr)
\bigl(\theta'-c_1 \tilde\eta r^{-1}\theta_0^{-\delta}\bigr)\Theta 
\bigl(A-\tfrac{\mathrm i}2|\mathrm dr|^2\theta'\bigr)
\\&\phantom{{}={}}
-\tfrac18c_1 \eta|\mathrm dr|^2r^{-1}\theta_0^{-\delta}\theta'^2\Theta 
+\tfrac18|\mathrm dr|^4\theta'\theta''\Theta
-C_5Q
.
\end{split}\label{eq:14.9.4.18.23ffaabbbbbbaae}
\end{align}
Substitute \eqref{eq:14.9.17.12.15ffaabbqq} and \eqref{eq:14.9.4.18.23ffaabbbbbbaae}
into \eqref{eq:14.9.4.18.22ffaabbqq}, 
and then it follows that 
\begin{align}
\begin{split}
\mathop{\mathrm{Im}}\bigl(A\Theta(H-\lambda)\bigr)
&\ge
c_1 (\lambda-q_1) r^{-1}\theta_0^{-\delta}\Theta 
+\tfrac18c_1\eta |\mathrm dr|^2r^{-1}\theta_0^{-\delta}\theta'^2\Theta
+\tfrac1{8}|\mathrm dr|^4\theta'\theta''\Theta
\\&\phantom{{}={}}
+\tfrac12
\bigl(A+\tfrac{\mathrm i}2|\mathrm dr|^2\theta'\bigr)
\bigl(\theta'-c_1 \tilde\eta r^{-1}\theta_0^{-\delta}\bigr)\Theta 
\bigl(A-\tfrac{\mathrm i}2|\mathrm dr|^2\theta'\bigr)
\\&\phantom{{}={}}
+\mathop{\mathrm{Re}}\Bigl[
\bigl(c_1 r^{-1}\theta_0^{-\delta}\Theta
-\tfrac12|\mathrm dr|^2\Theta'\bigr) (H-\lambda)\Bigr]
-C_6Q
.
\end{split}\label{eq:150916}
\end{align}
Using the formula \eqref{eq:13.9.7.9.32} we rewrite and bound the remainder operator
\eqref{eq:15091119} as
\begin{align}
\begin{split}
Q
&\le 
C_7(1+\alpha^2)r^{-1-\min\{1,\rho',\tau/2\}}\Theta
+C_7(1+\alpha^2)\bigl(\chi_{m-1,m+1}^2+\chi_{n-1,n+1}^2\bigr)r^{-1}\mathrm e^\theta
\\&\phantom{{}={}}
+2\mathop{\mathrm{Re}}\Bigl[\bigl(r^{-1-\min\{1,\rho',\tau/2\}}\chi_{m,n}
+|\chi_{m,n}'|\bigr)\mathrm e^\theta (H-\lambda)\Bigr]
.
\end{split}
\label{eq:15090983}
\end{align}
Hence we obtain by \eqref{eq:150916} and \eqref{eq:15090983}
\begin{align}
\begin{split}
&\mathop{\mathrm{Im}}\bigl(A\Theta(H-\lambda)\bigr)
\\&
\ge
\Bigl(c_1 (\lambda-q_1) r^{-1}\theta_0^{-\delta}
+\tfrac18c_1 |\mathrm dr|^2r^{-1}\theta_0^{-\delta}\theta'^2
+\tfrac18|\mathrm dr|^4\theta'\theta''
\\&\phantom{{}={}\Bigl(}
-C_8(1+\alpha^2)r^{-1-\min\{1,\rho',\tau/2\}}\Bigr)\Theta
\\&\phantom{{}={}}
+\tfrac12
\bigl(A+\tfrac{\mathrm i}2|\mathrm dr|^2\theta'\bigr)
\bigl(\theta'-c_1 \tilde\eta r^{-1}\theta_0^{-\delta}\bigr)\Theta 
\bigl(A-\tfrac{\mathrm i}2|\mathrm dr|^2\theta'\bigr)
\\&\phantom{{}={}}
-C_8(1+\alpha^2)\bigl(\chi_{m-1,m+1}^2+\chi_{n-1,n+1}^2\bigr)r^{-1}\mathrm e^\theta
+\mathop{\mathrm{Re}}\bigl(\gamma (H-\lambda)\bigr)
,
\end{split}\label{eq:14.9.4.18.23ffaabbbbggqq}
\end{align}
where 
\begin{align*}
\gamma
=
c_1 r^{-1}\theta_0^{-\delta}\Theta
-\tfrac12|\mathrm dr|^2\Theta'
-2C_6r^{-1-\min\{1,\rho',\tau/2\}}\Theta
-2C_6|\chi_{m,n}'|\mathrm e^\theta
.
\end{align*}

Now we restrict parameters.
Fix any $\alpha_0\ge 0$.
Then uniformly in $\alpha\in [0,\alpha_0]$
\begin{align*}
&c_1 (\lambda-q_1) r^{-1}\theta_0^{-\delta}
+\tfrac18c_1 |\mathrm dr|^2r^{-1}\theta_0^{-\delta}\theta'^2
+\tfrac18|\mathrm dr|^4\theta'\theta''
-C_8(1+\alpha^2)r^{-1-\min\{1,\rho',\tau/2\}}
\\&
\ge 
c_2r^{-1}\theta_0^{-\delta}
-C_9\beta r^{-1}\theta_0^{-1-\delta}
-C_9r^{-1-\min\{1,\rho',\tau/2\}},
\end{align*}
and hence, if we choose sufficiently small $\beta\in(0,1]$ and sufficiently large $n_0\ge 0$,
the first term of \eqref{eq:14.9.4.18.23ffaabbbbggqq} is bounded below
uniformly in $\alpha\in[0,\alpha_0]$, $n>m\ge n_0$ and $\nu\ge 0$ as
\begin{align*}
&\Bigl(
c_1 (\lambda-q_1) r^{-1}\theta_0^{-\delta}
+\tfrac18c_1 |\mathrm dr|^2r^{-1}\theta_0^{-\delta}\theta'^2
+\tfrac18|\mathrm dr|^4\theta'\theta''
-C_8(1+\alpha^2)r^{-1-\min\{1,\rho',\tau/2\}}
\Bigr)
\Theta
\\&
\ge c_3r^{-1}\theta_0^{-\delta}\Theta.
\end{align*}
Since 
\begin{align*}
\theta'-c_1 \tilde\eta r^{-1}\theta_0^{-\delta}
\ge 2\beta\theta_0^{-1-\delta}-C_{10}r^{-1}\theta_0^{-\delta},
\end{align*}
by retaking $n_0\ge 0$ larger, if necessary,
the second term of \eqref{eq:14.9.4.18.23ffaabbbbggqq} 
is non-negative for any $\alpha\in[0,\alpha_0]$, $n>m\ge n_0$ and $\nu\ge n_0$.
Hence the desired estimate \eqref{14.9.26.9.53ffaabb} follows.
\end{proof}

\begin{proof}[Proof of Proposition~\ref{prop:absence-eigenvalues-1b}]
Let $\lambda>\lambda_0$, $\phi\in  \mathcal H_{\mathrm{loc}}$ and $m_0\in \N$
be as in the assertion, and set
\begin{align*}
\alpha_0=\sup\{\alpha\ge 0\,|\,{\bar\chi_{m_0}}\mathrm e^{\alpha r}\phi\in B_0^*\}.
\end{align*}
Assume $\alpha_0<\infty$, and we shall find a contradiction.
Fix any $\delta\in (0,\min\{1,\rho',\tau/2\})$,
and choose $\beta>0$ and $n_0\ge 0$
in agreement with Lemma~\ref{lem:14.10.4.1.17ffaabb}.
Note that we may assume $n_0\ge m_0+3$,
so that for all $n>m\ge n_0$
$$\chi_{m-2,n+2}\phi\in \mathcal D(H).$$
If $\alpha_0=0$, let $\alpha=0$ so that automatically $\alpha+\beta>\alpha_0$.
Otherwise, we choose $\alpha\in[0,\alpha_0)$ such that $\alpha+\beta>\alpha_0$.
With these values of $\alpha$ and $\beta$
we take the expectation of both sides of the inequality \eqref{14.9.26.9.53ffaabb}
in the state $\chi_{m-2,n+2}\phi\in \mathcal D(H)$.
Noting that $\Theta(H-\lambda)\chi_{m-2,n+2}\phi=0$ and
$\gamma(H-\lambda)\chi_{m-2,n+2}\phi=0$  we obtain that  
for any $n>m\ge n_0$ and $\nu\ge n_0$
\begin{align}
\begin{split}
\bigl\|(r^{-1}\theta_0^{-\delta}\Theta)^{1/2}\phi\bigr\|^2
&
\le 
C_m\|\chi_{m-1,m+1}\phi\|^2
+C_\nu R_n^{-1}\|\chi_{n-1,n+1}\mathrm e^{\alpha r}\phi\|^2
.
\end{split}
\label{eq:11.7.16.3.22a}
\end{align}
The second term to the right of (\ref{eq:11.7.16.3.22a})
  vanishes when  $n\to\infty$ since 
$\bar\chi_{m_0}\mathrm e^{\alpha r}\phi\in B^*_0$, 
and consequently by Lebesgue's monotone  convergence
theorem
\begin{align}
\bigl\|(\bar\chi_m r^{-1}\theta_0^{-\delta}
\mathrm{e}^{\theta})^{1/2}\phi\bigr\|^2
 &\le 
C_m\|\chi_{m-1,m+1}\phi\|^2.
\label{eq:11.7.16.3.43a}
\end{align}
Next we let $\nu \to\infty$ in \eqref{eq:11.7.16.3.43a}
    invoking again Lebesgue's monotone  convergence
theorem,
and then it follows that 
 $$\bar\chi_m^{1/2} r^{-1/2}\mathrm e^{(\alpha+\beta) r}\phi\in \mathcal H.$$ 
Consequently this implies $\bar\chi_m^{1/2}\mathrm e^{\kappa r}\phi\in B_0^*$
for any $\kappa\in(0,\alpha+\beta)$. 
But this is a contradiction, since $\alpha+\beta>\alpha_0$. 
We are done.
\end{proof}

\subsection{Absence of super-exponentially decaying eigenstates}\label{subsec:150909513}
In this subsection we prove Proposition ~\ref{prop:absence-eigenvalues-1bb}.
{A remark similar to Remark~\ref{rem:190802} applies to the following.}

\begin{lemma}\label{lem:14.10.4.1.17ffaabbqqq}
Let $\lambda>\lambda_0$ and $\alpha_0>0$, 
and fix $\beta=0$ in the definition \eqref{eq:15.2.15.5.8bb} of $\Theta$,
i.e.\ $\Theta=\chi_{m,n}\mathrm e^{2\alpha r}$.
Then there exist $c,C>0$ and $n_0\ge 0$ such that uniformly in 
$\alpha>\alpha_0$ and $n>m\ge n_0$,  
as quadratic forms on $\mathcal D(H)$,
\begin{align}
\begin{split}
\mathop{\mathrm{Im}}\bigl(A\Theta (H-\lambda)\bigr)
&\ge 
c\alpha^2 r^{-1}\Theta
-C\alpha^2\bigl(\chi_{m-1,m+1}^2+\chi_{n-1,n+1}^2\big)r^{-1}\mathrm e^{2\alpha r}
\\&\phantom{={}}
+\mathop{\mathrm{Re}}\bigl(\gamma(H-\lambda)\bigr),
\end{split}
\label{14.9.26.9.53ffaabbqqqq}
\end{align}
where $\gamma=\gamma_{m,n}$ is a certain function
satisfying $\mathop{\mathrm{supp}}\gamma\subseteq\mathop{\mathrm{supp}}\chi_{m,n}$ 
and {$|\gamma|\le C_{m,n}\alpha\mathrm e^{2\alpha r}$.}
\end{lemma}
\begin{proof}
Fix any $\lambda>\lambda_0$ and $\delta\in (0,\min\{1,\rho',\tau/2\})$.
Then, repeating the arguments of the proof of Lemma~\ref{lem:14.10.4.1.17ffaabb},
we have the bound \eqref{eq:14.9.4.18.23ffaabbbbggqq} 
uniformly in $\alpha\ge 0$, $\beta\in[0,1]$, $n>m\ge 0$ and $\nu\ge 0$.
There we fix any $\alpha_0>0$,
let $\beta=0$ and $\nu\to\infty$, 
and choose sufficiently large $n_0\ge 0$.
Consequently we can easily verify
the asserted inequality \eqref{14.9.26.9.53ffaabbqqqq} 
uniformly in $\alpha>\alpha_0$ and $n>m\ge n_0$. 
Hence we are done.
\end{proof}

\begin{proof}[Proof of Proposition~\ref{prop:absence-eigenvalues-1bb}]
Let $\lambda>\lambda_0$, $\phi\in \mathcal H_{\mathrm{loc}}$ and $m_0\in \N$
be fixed as in the proposition.
Fix any $\alpha_0>0$ and $\beta=0$, 
and choose  $n_0\ge 0$ in agreement with
Lemma~\ref{lem:14.10.4.1.17ffaabbqqq}. We may assume that $ n_0\geq m_0+3$,
so that for all $n>m\ge n_0$
\begin{align*}
\chi_{m-2,n+2}\phi\in\mathcal D(H).
\end{align*}
Let us evaluate the inequality \eqref{14.9.26.9.53ffaabbqqqq} 
in the state $\chi_{m-2,n+2}\phi\in\mathcal D(H)$.
Then it follows that for any $\alpha>\alpha_0$ and $n>m\ge n_0$
\begin{align}
\|r^{-1/2}\Theta^{1/2}\phi\|^2
\leq C_1\|\chi_{m-1,m+1}\mathrm e^{\alpha r}\phi\|^2
+C_1 R_n^{-1}\|\chi_{n-1,n+1}\mathrm e^{\alpha r}\phi\|^2.
\label{eq:11.10.23.12.54}
\end{align}
The second term to the right of \eqref{eq:11.10.23.12.54} 
vanishes when $n\to\infty$, 
and hence by Lebesgue's monotone  convergence
theorem we obtain
\begin{align*}
\bigl\|\bar\chi_m^{1/2} r^{-1/2}\mathrm e^{\alpha r}\phi\bigr\|^2
&
\le 
C_1\|\chi_{m-1,m+1}\mathrm e^{\alpha r}\phi\|^2,
\end{align*}
or 
\begin{align}
\bigl\|\bar\chi_m^{1/2} r^{-1/2}\mathrm e^{\alpha (r-4R_m)}\phi\bigr\|^2
&
\le 
C_1\|\chi_{m-1,m+1}\phi\|^2
.
\label{eq:11.7.16.3.43qq}
\end{align}

Now assume $\bar\chi_{m+2}\phi\not\equiv 0$.
The left-hand side of \eqref{eq:11.7.16.3.43qq}
 grows exponentially as $\alpha\to\infty$
whereas the right-hand side remains bounded.
This is a contradiction.
Thus $\bar\chi_{m+2}\phi\equiv 0$.
By invoking the unique continuation property for the second order
elliptic operator $H$ we conclude that $\phi\equiv 0$ globally on $M$. 
We refer to \cite{Wo} for the unique continuation property for $d\ge 2$.
The case $d=1$ is due to uniqueness of a solution to an ordinary differential equation.
\end{proof}

\section{Besov bound}\label{sec:besov-bound-resolv}

In this section we discuss the locally uniform Besov bound
for the resolvent $R(z)$, and prove Theorem~\ref{thm:12.7.2.7.9}.
A key to the proof is a kind of single commutator
estimate with weight inside stated in Lemma~\ref{lem:14.10.4.1.17ffaa},
or its direct consequence, Proposition~\ref{prop:12.7.2.7.9}.

Here let us state a slightly technical but main proposition of the section
that follows from  Lemma~\ref{lem:14.10.4.1.17ffaa}.
We introduce the regularized weight
\begin{align}
\Theta=\Theta_\nu^\delta=\int_0^{r/R_\nu}(1+s)^{-1-\delta}\,\mathrm ds
=\bigl[1-(1+r/R_\nu)^{-\delta}\bigr]\big/\delta
;\quad \delta>0,\  \nu\ge 0,
\label{eq:15.2.15.5.8}
\end{align}
and compute   derivatives in $r$:
\begin{align}
\Theta'=(1+r/R_\nu)^{-1-\delta}\big/R_\nu,\quad
\Theta''=-(1+\delta)(1+r/R_\nu)^{-2-\delta}\big/R_\nu^2.
\label{eq:15.2.15.5.9}
\end{align}
Recall the notation defined right before Theorem~\ref{thm:12.7.2.7.9}.

\begin{proposition}\label{prop:12.7.2.7.9}
Suppose Conditions~\ref{cond:12.6.2.21.13}--\ref{cond:10.6.1.16.24} 
let $I\subseteq \mathcal I$ be a compact interval,
and fix any $\delta\in (0, \min\{1,\rho',\tau/2\})$ 
in the definition \eqref{eq:15.2.15.5.8} of $\Theta$.
Then
there exist $C>0$ and $n\ge 0$ such that 
for all $\phi=R(z)\psi$ with $z\in I_\pm$ and $\psi\in B$ and for all $\nu\ge 0$ 
\begin{align}
\begin{split}
&\|\Theta'{}^{1/2}\phi\|^2
+\|\Theta'{}^{1/2}A\phi\|^2
+\langle p_i^*\Theta h^{ij}p_j\rangle_\phi
\\&\le C\Bigl(\|\phi\|_{B^*}\|\psi\|_B
+\|A\phi\|_{B^*}\|\psi\|_B
+\|\chi_n \Theta^{1/2}\phi\|^2\Bigr).
\end{split}\label{eq:13.8.22.4.59cc}
\end{align}
\end{proposition}

In Subsection~\ref{subsec:Improved radiation conditionsb}
we prove Lemma~\ref{lem:14.10.4.1.17ffaa}
and, as a result, Proposition~\ref{prop:12.7.2.7.9} too.
In Subsection~\ref{subsec:15.2.14.14.41},
combining Proposition~\ref{prop:12.7.2.7.9} and Condition ~\ref{cond:12.6.2.21.13b},
we prove Theorem~\ref{thm:12.7.2.7.9} by contradiction.
%We remark that the contradiction argument can be partially avoided
%if we assume more regularity on $r$ and $q$, and employ a double commutator estimate, 
%although we do not discuss it in the present paper.

\subsection{Commutator estimate}\label{subsec:Improved radiation conditionsb}

We first note properties of $\Theta$
defined by \eqref{eq:15.2.15.5.8}.
\begin{lemma}\label{lem:15.2.18.14.28}
Suppose Condition~\ref{cond:12.6.2.21.13},
and fix any $\delta>0$
in the definition \eqref{eq:15.2.15.5.8} of $\Theta$.
Then there exist $c,C,C_k>0$, $k=2,3,\dots$, such 
that for any $k=2,3,\dots$ and uniformly in $\nu\ge 0$  
\begin{align*}
&c/R_\nu \le \Theta\le \min\{C,r/R_\nu\},\\
&c \parb{\min\{R_\nu,r\}}^{\delta}r^{-1-\delta}\Theta\le \Theta'\le r^{-1}\Theta,\\
&0\le (-1)^{k-1}\Theta^{(k)}\le C_kr^{-k}\Theta.
%\label{eq:14.10.6.1.45ww}
\end{align*}
\end{lemma}
\begin{proof}
All the asserted estimates are straightforward except possibly for the
first estimate in the second line. But by using the last estimate of the
first line this estimate follows if we can bound
\begin{align*}
  \parb{\min\{R_\nu,r\}}^{\delta}r^{-1-\delta}\parb{\min\{R_\nu,r\}/R_\nu}\parb{(1+r/R_\nu)^{1+\delta}R_\nu}\leq C,
\end{align*} which  obviously is correct for the case $r\leq R_\nu$ as
well as for the case $r> R_\nu$.
\end{proof}

Now we state and prove our key lemma:
\begin{lemma}\label{lem:14.10.4.1.17ffaa}
Suppose Conditions~\ref{cond:12.6.2.21.13}--\ref{cond:10.6.1.16.24} 
let $I\subseteq \mathcal I$ be a compact interval,
and fix any $\delta\in (0, \min\{1,\rho',\tau/2\})$ 
in the definition \eqref{eq:15.2.15.5.8} of $\Theta$.
Then there exist $c,C>0$ and $n\ge 0$ such that uniformly in $z\in I_\pm$
and $\nu\ge 0$, 
as quadratic forms on $\mathcal D(H)$,
\begin{align}
\begin{split}
\mathop{\mathrm{Im}}
\bigl(A\Theta(H-z)\bigr)
&\ge c\Theta'+cA\Theta'A+cp_i^*\Theta h^{ij}p_j
-C\chi_n^2\Theta-\mathop{\mathrm{Re}}\bigl(\gamma (H-z)\bigr),
\end{split}
\label{14.9.26.9.53ffaa}
\end{align}
where $\gamma=\gamma_{z,\nu}$ is a uniformly bounded complex-valued function:
$|\gamma|\le C$.
\end{lemma}
\begin{proof}
Let $I$ and $\delta$ be as in the assertion.
We are going to prove the lemma 
by bounding the form on the left-hand side of \eqref{14.9.26.9.53ffaa}.
First by Lemmas~\ref{lem:undoing-commutators}, \ref{lem:12.7.3.0.45}, 
\ref{lem:15.2.18.14.28}, \eqref{eq:12.6.27.0.43}, the Cauchy--Schwarz inequality
and \eqref{eq:14.12.27.3.22}
we can bound uniformly in $z=\lambda\pm\mathrm i\Gamma\in I_\pm$ and $\nu\ge 0$
\begin{align}
\begin{split}
\mathop{\mathrm{Im}}\bigl(A\Theta(H-z)\bigr)
&
\ge 
\tfrac12A\Theta' A
+\tfrac12p_i^*\Theta h^{ik}p_k
-\tfrac12\mathop{\mathrm{Im}}\bigl(\Theta'(\mathrm dr)_ih^{ij}p_j\bigr)
\\&\phantom{{}={}}
\mp\Gamma\mathop{\mathrm{Re}}(A\Theta)
-\tfrac12\mathop{\mathrm{Re}}\bigl(|\mathrm dr|^2\Theta'(H-\lambda)\bigr)
-C_1Q
\\&
\ge 
\tfrac12c_1A\tilde\eta\Theta' A
+\tfrac12c_1p_i^*\Theta'\ell^{ik}p_k
\\&\phantom{{}={}}
+\tfrac12A\bigl(1-c_1\tilde\eta\bigr)\Theta' A
+\tfrac14p_i^*\bigl(\Theta h^{ik}-2c_1\Theta'\ell^{ik}\bigr)p_k
\\&\phantom{{}={}}
\mp\Gamma\Theta^{1/2}A\Theta^{1/2}
-\tfrac12\mathop{\mathrm{Re}}\bigl(|\mathrm dr|^2\Theta'(H-z)\bigr)
-C_2Q
,
\end{split}\label{eq:14.9.4.18.22ffaa}
\end{align} 
where $c_1>0$ is a small constant and 
\begin{align*}
Q=r^{-1-\min\{1,\rho',\tau/2\}}\Theta
+p_i^*r^{-1-\min\{1,\rho',\tau/2\}}\Theta g^{ij}p_j.
\end{align*}

We further estimate the terms on the right-hand side
of \eqref{eq:14.9.4.18.22ffaa} as follows.
By Lemma~\ref{lem:15.2.18.14.28}
we can choose and fix $c_1>0$ in \eqref{eq:14.9.4.18.22ffaa} small enough that
the third and fourth terms on the right-hand side satisfy
\begin{align}
\tfrac12A\bigl(1-c_1\tilde\eta\bigr)\Theta' A
+\tfrac14p_i^*\bigl(\Theta h^{ik}-2c_1\Theta'\ell^{ik}\bigr)p_k
\ge 
c_2A\Theta'A+c_2p_i^*\Theta h^{ik}p_k.
\label{eq:14.9.17.12.15ffaa}
\end{align}
Using Lemma~\ref{lem:15.2.18.14.28}, the Cauchy-Schwarz inequality and \eqref{eq:15091010}
we can rewrite and bound the first and second terms of \eqref{eq:14.9.4.18.22ffaa}
as 
\begin{align}
\begin{split}
\tfrac12c_1A\tilde\eta  \Theta' A
+\tfrac12c_1p_i^*\Theta' \ell^{ij}p_j
&
\ge 
\tfrac12c_1\mathop{\mathrm{Re}}
\Bigl[
\Theta' A\tilde\eta  A
+\Theta'  L \Bigr]
-C_3Q
\\&
\ge 
c_1(\lambda-q_1) \Theta' 
+\mathop{\mathrm{Re}}
\bigl(c_1\Theta' (H-z)\bigr)
-C_4Q
.
\end{split}\label{eq:15.2.18.16.9}
\end{align}
To the fifth term of \eqref{eq:14.9.4.18.22ffaa}
we apply the Cauchy-Schwarz inequality,
Lemma~\ref{lem:15.2.18.14.28}
and the general identity 
holding for any real functions $f,h\in C^1(M)$:
\begin{align}\label{eq:commform}
  h\Re (fH_0)h = \Re (h^2fH_0) +\tfrac 12(\partial_i h)
  g^{ij}(\partial_j hf).
\end{align}
Then it follows that
\begin{align}
\begin{split}
\mp\Gamma \Theta^{1/2} A\Theta^{1/2}
&\ge 
-C_5\Gamma \Theta^{1/2} (H-\lambda)\Theta^{1/2}
-
C_6\Gamma 
\\&\ge 
-C_5\Gamma \mathop{\mathrm{Re}}\bigl(\Theta^{1/2}(H-z)\Theta^{1/2}\bigr)
\pm C_6\mathop{\mathrm{Im}}(H-z)
\\&\ge 
-C_7Q
-\mathop{\mathrm{Re}}\bigl((C_5\Gamma\Theta\pm \mathrm iC_6) (H-z)\bigr)
.
\end{split}
\label{eq:14.10.3.2.16ffaa}
\end{align}

We substitute the estimates 
\eqref{eq:14.9.17.12.15ffaa},
\eqref{eq:15.2.18.16.9}
and 
\eqref{eq:14.10.3.2.16ffaa}
 into 
\eqref{eq:14.9.4.18.22ffaa},
and obtain
\begin{align}
\begin{split}
\mathop{\mathrm{Im}}\bigl(A\Theta(H-z)\bigr)
&\ge
c_3\Theta'
+c_2 A\Theta'A
+c_2 p_i^*\Theta h^{ij}p_j
-C_8Q
\\&\phantom{{}={}}
-\mathop{\mathrm{Re}}
\Bigl(\bigl[\bigl(\tfrac12|\mathrm dr|^2-c_1\bigr)\Theta'
+C_5\Gamma\Theta\pm \mathrm iC_6
\bigr](H-z)\Bigr)
.
\end{split}\label{eq:14.9.4.18.23ffaa}
\end{align}
Finally we can use \eqref{eq:13.9.7.9.32} and Lemma~\ref{lem:15.2.18.14.28}
to combine and estimate the first and fourth terms of \eqref{eq:14.9.4.18.23ffaa}: 
For small $c_4>0$ and large $n\ge 0$
\begin{align}
\begin{split}
c_3\Theta'
-C_8Q
&
\ge 
c_4\Theta'
-C_9\chi_n^2\Theta
-2C_8 \mathop{\mathrm{Re}}
\bigl(r^{-1-\min\{1,\rho',\tau/2\}}\Theta(H-z)\bigr)
.
\end{split}
\label{eq:14.9.17.12.17ffaa}
\end{align}
Hence by \eqref{eq:14.9.4.18.23ffaa} and \eqref{eq:14.9.17.12.17ffaa},
if we set 
\begin{align*}
\gamma=\bigl(\tfrac12|\mathrm dr|^2-c_1\bigr)\Theta'
+C_5\Gamma\Theta\pm \mathrm iC_6
+2C_8r^{-1-\min\{1,\rho',\tau/2\}}\Theta,
\end{align*}
then the desired inequality \eqref{14.9.26.9.53ffaa} follows.
\end{proof}

\begin{proof}[Proof of Proposition~\ref{prop:12.7.2.7.9}]
The assertion follows immediately from Lemma~\ref{lem:14.10.4.1.17ffaa}.
\end{proof}

\subsection{Compactness and contradiction}\label{subsec:15.2.14.14.41}
Now we suppose Condition ~\ref{cond:12.6.2.21.13b}, 
and prove Theorem~\ref{thm:12.7.2.7.9} by Proposition~\ref{prop:12.7.2.7.9} and contradiction.
\begin{proof}[Proof of Theorem~\ref{thm:12.7.2.7.9}]
Let $I\subseteq \mathcal I$ be a compact interval.
We prove the assertion only for the upper sign.

\smallskip
\noindent
\textit{Step I.}
We first reduce the bound \eqref{eq:13.8.22.4.59c}
to the single bound
\begin{align}
\|\phi\|_{B^*}\le C_1\|\psi\|_B.
\label{eq:15.2.15.5.2}
\end{align}
In fact, assume \eqref{eq:15.2.15.5.2}.
Then the last term of the left-hand side of \eqref{eq:13.8.22.4.59c}
clearly satisfies the desired estimate by the identity
\begin{align*}
H_0\phi=\psi-(V-z)\phi
\end{align*}
and the fact that $V$ is bounded by Conditions~\ref{cond:12.6.2.21.13}--\ref{cond:10.6.1.16.24}.
Hence it suffices to consider the second and third terms of  \eqref{eq:13.8.22.4.59c}.
Fix any $\delta\in (0, \min\{1,\rho',\tau/2\})$.
Then by Proposition~\ref{prop:12.7.2.7.9} and \eqref{eq:15.2.15.5.2}
there exists $C_2>0$ such that
for any $\phi=R(z)\psi$ with $z\in I_+$ and $\psi\in B$
uniformly in $\epsilon_1\in (0,1)$ and $\nu\ge 0$
\begin{align}
\begin{split}
&
\|\Theta'{}^{1/2}A\phi\|^2
+\langle p_i^*\Theta h^{ij}p_j\rangle_\phi
\le \epsilon_1^{-1}C_2\|\psi\|_B^2
+\epsilon_1 \|A\phi\|_{B^*}^2
.
\end{split}\label{eq:13.8.22.4.59ccd}
\end{align}
In the first term on the left-hand side of \eqref{eq:13.8.22.4.59ccd}
for each $\nu\ge 0$,
noting the expression of $\Theta'$ in \eqref{eq:15.2.15.5.9},
we restrict the integral region to $B_{R_{\nu+1}}\setminus B_{R_\nu}$.
As for the second term on the same side
we look at the estimate \eqref{eq:13.8.22.4.59ccd} 
for any fixed $\nu\ge 0$, say $\nu=0$.
Then we can deduce from \eqref{eq:13.8.22.4.59ccd} that 
\begin{align*}
\begin{split}
&
c_1\|A\phi\|_{B^*}^2
+c_1\langle p_i^*h^{ij}p_j\rangle_\phi
\le 2\epsilon_1^{-1}C_2\|\psi\|_B^2
+2\epsilon_1 \|A\phi\|_{B^*}^2
.
\end{split}%\label{eq:13.8.22.4.59ccdd}
\end{align*}
If we let $\epsilon_1\in (0,c_1/2)$, the rest of \eqref{eq:13.8.22.4.59c} follows
from this estimate and \eqref{eq:12.6.27.0.43}.
Hence \eqref{eq:13.8.22.4.59c}
reduces to \eqref{eq:15.2.15.5.2}.

\smallskip
\noindent
\textit{Step I\hspace{-.1em}I.}
We prove \eqref{eq:15.2.15.5.2} by contradiction.
Assume the opposite, and let $z_k\in I_+$ and $\psi_k\in B$ be such that 
\begin{align}
\lim_{k\to\infty}\|\psi_k\|_B=0,\quad
\|\phi_k\|_{B^*}=1;\quad
\phi_k=R(z_k)\psi_k.
\label{eq:15.2.15.5.51}
\end{align}
Note that then it automatically follows that 
\begin{align}
\|p\phi_k\|_{B^*}+\|H_0\phi_k\|_{B^*}\le C_3.\label{eq:15.2.15.14.0}
\end{align}
In fact, arguing similarly to Step I,
we can deduce from \eqref{eq:15.2.15.5.51} and Proposition~\ref{prop:12.7.2.7.9} 
that 
\begin{align*}
\begin{split}
\|A\phi_k\|_{B^*}^2
+\langle p_i^*h^{ij}p_j\rangle_{\phi_k}
\le C_4,\quad
H_0\phi_k
=\psi_k-(V-z_k)\phi_k
,
\end{split}%\label{eq:13.8.22.4.59ccde}
\end{align*}
and these combined Condition~\ref{cond:12.6.2.21.13a}, \eqref{eq:12.6.27.0.43} 
and \eqref{eq:15.2.15.5.51} imply \eqref{eq:15.2.15.14.0}.
Now, choosing a subsequence we may assume that $z_k\to z\in I$.
In fact if the imaginary part of the   limit $z$ of a convergent
subsequence is positive, the bounds
\begin{align*}
\|\phi_k\|_{B^*}
\le 
\|\phi_k\|_{\mathcal H}
\le \|R(z_k)\|_{\mathcal B(\mathcal H)}\|\psi_k\|_{\mathcal H}
\le C_5\|R(z_k)\|_{\mathcal B(\mathcal H)}\|\psi_k\|_B
\end{align*}
and \eqref{eq:15.2.15.5.51} contradict the norm continuity of $R(z)\in
\mathcal B(\mathcal H)$. Hence indeed we have the limit 
\begin{align}
\lim_{k\to\infty}z_k=z=\lambda\in I.
\label{eq:15.2.15.5.52}
\end{align}
Let $s>1/2$.  By choosing a further subsequence we may assume that 
$\phi_k$ converges  weakly  to some $\phi$ in $\mathcal H_{-s}$, cf. \cite{Yo}.
But then $\phi_k$ actually converges strongly in $\mathcal H_{-s}$.
To see this let us fix $s'\in (1/2,s)$ and $f\in C^\infty_\c(\mathcal I)$ 
with $f=1$ on a neighborhood of $I$, and decompose
for any  $n\ge 0$
\begin{align*}
r^{-s}\phi_k
&=r^{-s}f(H)(\chi_{n}r^s)(r^{-s}\phi_k)
+(r^{-s}f(H)r^s)(\bar\chi_nr^{s'-s})( r^{-s'}\phi_k)
\\&\phantom{{}={}}
+r^{-s}(1-f(H))R(z_k)\psi_k.
\end{align*}
The last term on the right-hand side converges to $0$ 
in $\mathcal H$ due to \eqref{eq:15.2.15.5.51}, 
and the second term can be taken arbitrarily small in $\mathcal H$ 
by choosing $n\ge 0$ sufficiently large 
since $r^{-s}f(H)r^s$ is a bounded operator.
  It follows from  Condition~\ref{cond:12.6.2.21.13b} that $r^{-s}f(H)$
is compact. Whence
for  fixed $n\ge 0$ the first term converges strongly in $\mathcal H$.
 So  indeed $\phi_k$ converges to $\phi$ strongly in $\mathcal H_{-s}$.
By using \eqref{eq:15.2.15.14.0}
we can see that $p\phi\in \mathcal H_{-s}$, 
and then by using \eqref{eq:13.9.7.9.32}, 
or alternatively the first resolvent equation, 
we can see that the sequence $\{p\phi_k\}$ is a Cauchy sequence in 
$\mathcal H_{-s}$. Whence we have 
\begin{align}
\lim_{k\to\infty}\phi_k=\phi\quad \text{in }\mathcal H_{-s},\quad
\lim_{k\to\infty}p\phi_k=p\phi\quad \text{in }\mathcal H_{-s}.
\label{eq:15.2.15.5.53}
\end{align}
By \eqref{eq:15.2.15.5.51}, \eqref{eq:15.2.15.5.52} and \eqref{eq:15.2.15.5.53}
it follows that 
\begin{align}
\phi\in\mathcal N,\quad
(H-\lambda)\phi=0 \text{ in the distributional sense}
.
\label{eq:15.2.15.14.32}
\end{align}
In addition, we can verify $\phi\in B^*_0$.
In fact, let us apply Proposition~\ref{prop:12.7.2.7.9} with $\delta=2s-1>0$
to $\phi_k=R(z_k)\psi_k$,
and take the limit $k\to\infty$ using \eqref{eq:15.2.15.5.53},
\eqref{eq:15.2.15.5.51}, \eqref{eq:15.2.15.14.0} and Lemma
\ref{lem:15.2.18.14.28}.
We obtain for all $\nu\ge 0$
\begin{align}
\begin{split}
\|\Theta'{}^{1/2}\phi\|
&\le C_6\|\chi_n\Theta^{1/2}\phi\|
\le C_6R_\nu^{-1/2} \|\chi_nr^{1/2}\phi\|
.
\end{split}
\label{eq:15.2.15.12.9}
\end{align}
Letting $\nu\to\infty$ in \eqref{eq:15.2.15.12.9},
we obtain 
$\phi\in B^*_0$,
and then we conclude $\phi=0$ by \eqref{eq:15.2.15.14.32} and Theorem~\ref{thm:13.6.20.0.10}.
But this is a contradiction, because similarly to Step I
we have 
\begin{align*}
1=\|\phi_k\|_{B^*}^2\le C_7\bigl(\|\psi_k\|_B+\|\chi_n\phi_k\|^2\big),
\end{align*}
and, as $k\to \infty$, the right-hand side converges to $0$.
Hence \eqref{eq:15.2.15.5.2} holds.
\end{proof}

\section{Radiation condition}\label{sec:Radiation conditions
  -preliminary}

In this section we discuss the radiation condition bounds and their relevant consequences.
In Subsection~\ref{subsec:Improved radiation conditionsbb} we state and prove 
the main key commutator estimate of the section,
which is somewhat similar to that of Section~\ref{sec:besov-bound-resolv}.
In Subsection~\ref{subsec:15.3.9.14.28}, using this key estimate, 
we prove Theorem~\ref{prop:radiation-conditions}.
Corollaries~\ref{thm:12.7.2.7.9b}--\ref{thm:13.9.9.8.23} are also proved in the same subsection.

Throughout the section we suppose Condition~\ref{cond:12.6.2.21.13bbb},
and prove the statements only for the upper sign for simplicity.

\subsection{Commutator estimate}\label{subsec:Improved radiation conditionsbb}

\begin{lemma}\label{lem:13.9.2.7.18}
Let $I\subseteq \mathcal I$ be a compact interval.
There exists $C>0$ such that uniformly in $z\in I\cup I_+$
\begin{align*}
%\begin{split}
& |a|\le C,\quad
\mathop{\mathrm{Im}}a\ge -Cr^{-1-\min\{\rho',\rho/2\}},
\quad
q_{21}\mathop{\mathrm{Im}}a\ge -C(\Gamma+r^{-1})r^{-\rho},
\\
&\bigl|p^ra+a^2-2|\mathrm dr|^2(z-q_1)\bigr|
+\bigl|\ell^{\bullet i}\nabla_i a\bigr|
\le Cr^{-1-\min\{\rho/2,\tau/2\}}
.
%\end{split}
%\label{eq:13.9.2.7.18}
\end{align*}
\end{lemma}
\begin{proof}
It is clear by the definition \eqref{eq:13.9.5.7.23} that the function $a$ is bounded.
For the second estimate it suffices to note that 
by Condition~\ref{cond:12.6.2.21.13bbb}
\begin{align*}
\nabla^rq_{11}=\nabla^rq_1-\nabla^rq_{12}\le C_1r^{-1-\min\{\rho',\rho/2\}}.
\end{align*}
The third estimate is also clear by Condition~\ref{cond:12.6.2.21.13bbb}.
Since we can write
\begin{align*}
&p^ra+a^2-2|\mathrm dr|^2(z-q_1)\\
&=
(p^r\eta_\lambda|\mathrm dr|)\sqrt{2(z-q_1)}
+\tfrac14(p^r\eta_\lambda p^rq_{11})\big/(z-q_1)
\\&\phantom{{}={}}
+\tfrac14\eta_\lambda\bigl(1+\tfrac14\eta_\lambda\bigr)(p^rq_{11})^2\big/(z-q_1)^2
+\tfrac14\eta_\lambda(p^rq_{11})(p^rq_{12})\big/(z-q_1)^2
\\&\phantom{{}={}}
-\eta_\lambda|\mathrm dr|\bigl(p^rq_1-\eta_\lambda p^rq_{11}\bigr)\big/\sqrt{2(z-q_1)}
-2(1-\eta_\lambda^2)|\mathrm dr|^2(z-q_1),
\end{align*}
it is clear that this quantity satisfies the
assertion by Condition~\ref{cond:12.6.2.21.13bbb}.
The last bound is also clear.
\end{proof}

To simplify a commutator computation in Lemma~\ref{lem:14.10.22.18.8ff}
we decompose $H-z$ into a sum of  radial and  spherical components
like  \eqref{eq:15091010}.
\begin{lemma}\label{lem:15.1.15.15.16ff}
Let $I\subseteq \mathcal I$ be a compact interval.
Then there exist a complex-valued function $q_3$ and a constant $C>0$ such that 
uniformly in $z\in I\cup I_+$,
as quadratic forms on $\mathcal H^1$,
\begin{align*}
\begin{split}
H-z
&=
\tfrac 12(A+a)\tilde\eta (A-a)+\tfrac12 L  +q_{21}+q_3;\quad 
|q_3|\le Cr^{-1-\min\{\rho/2,\tau/2\}}
.
\end{split}%\label{eq:15.1.16.12.45}
\end{align*}
\end{lemma}
\begin{proof}
Using the expression \eqref{eq:15091010} we can write
\begin{align*}
\begin{split}
H-z
&=
\tfrac 12(A+a)\tilde\eta (A-a)
+\tfrac12(p^r\tilde\eta a)
+\tfrac12\tilde\eta a^2
+\tfrac12L  
+q_1+q_2+\tfrac14(\nabla^r\tilde\eta)(\Delta r)
-z
.
\end{split}%\label{eq:15.1.16.12.45aa}
\end{align*}
Hence the desired identity is obtained by setting 
\begin{align*}
q_3&=\tfrac 12\tilde\eta\bigl[
(p^ra)+a^2
-2|\mathrm dr|^2(z-q_1)
\bigr]
-(1-\eta)(z-q_1)
\\&\phantom{{}=}
+q_{22}
-\tfrac{\mathrm i}2(\nabla^r\tilde\eta )a
+\tfrac14(\nabla^r\tilde\eta)(\Delta r)
\end{align*}
and applying Lemma~\ref{lem:13.9.2.7.18}.
\end{proof}
\begin{remark}\label{rem:15.3.12.14.9}
The decay rate of $q_3$ depends very much on how accurately we can construct an
approximate solution to the radial Riccati equation \eqref{eq:15.3.11.19.35}.
We note that the change of variable $a=\pm (p^rb)/b$
reduces the equation \eqref{eq:15.3.11.19.35} to the second-order linear differential equation
\begin{align*}
(p^r)^2b-2|\mathrm dr|^2(z-q_1)b=0,
\end{align*}
which is nothing but a one-dimensional Schr\"odinger eigenequation with 
a long-range perturbation.
\end{remark}

Next we state and prove the key commutator estimate of the section
that is needed for our proof of the radiation condition bounds.
Let us introduce the regularized weight
\begin{align*}
\Theta=\Theta_\nu^\delta=\int_0^{r/R_\nu}(1+s)^{-1-\delta}\,\mathrm ds
=\bigl[1-(1+r/R_\nu)^{-\delta}\bigr]\big/\delta
;\quad  \delta>0,\ \nu\ge 0,
\end{align*}
which is the same weight as \eqref{eq:15.2.15.5.8} introduced in
Section~\ref{sec:besov-bound-resolv}.
We denote its derivatives in $r$ by primes such as \eqref{eq:15.2.15.5.9}.
Again we shall use Lemma~\ref{lem:15.2.18.14.28}  although we are going to choose
$\delta>0$ differently from Section~\ref{sec:besov-bound-resolv}.

\begin{lemma}\label{lem:14.10.22.18.8ff}
Let $I\subseteq \mathcal I$ be a compact interval
and fix any $\delta\in (0,\min\{\rho'/2,\rho/4,\tau/4\}]$ 
and $\beta\in (0,\sigma/2)$. 
Then there exist $c,C>0$ such that uniformly in $z\in I\cup I_+$ and $\nu\ge 0$,
as quadratic forms on  $\mathcal D(H)$
\begin{align*}
\begin{split}
\mathop{\mathrm{Im}}
\bigl((A-a)^*\Theta^{2\beta}(H-z)\bigr)
&\ge c(A-a)^*\Theta'\Theta^{2\beta-1}(A-a)
+cp_i^*\Theta^{2\beta}h^{ij}p_j
\\&\phantom{{}={}}
-Cr^{-1-\min\{\rho,\tau\}+2\delta}\Theta^{2\beta}
-\mathop{\mathrm{Re}}\bigl(\gamma\Theta^{2\beta}(H-z)\bigr),
\end{split}
%\label{14.9.26.9.53ff}
\end{align*}
where $\gamma$ is a complex-valued function satisfying 
$|\gamma|\le Cr^{-\min\{\rho,\tau\}+2\delta}$.
\end{lemma}
\begin{proof}
Let $I$, $\delta$ and $\beta$ be as in the statement,
and we are going to expand and bound the left-hand side of the asserted inequality.
By Lemma~\ref{lem:undoing-commutators}
we may compute it formally on $C^\infty_{\mathrm c}(M)$, as long as we bound it below.
For the practical computations below we proceed similarly to the proof of 
Lemma~\ref{lem:12.7.3.0.45}
employing Lemma~\ref{lem:15.1.15.15.16ff} instead of \eqref{eq:15091010}.
By Lemmas~\ref{lem:15.1.15.15.16ff}, \ref{lem:13.9.2.7.18} and
\ref{lem:15.2.18.14.28} and the Cauchy--Schwarz inequality
it follows that uniformly in $z\in I\cup I_+$ and $\nu\ge 0$
\begin{align}
\begin{split}
&\mathop{\mathrm{Im}}\bigl((A-a)^*\Theta^{2\beta}(H-z)\bigr)\\
&=
\tfrac12\mathop{\mathrm{Im}}\bigl((A-a)^*\Theta^{2\beta}(A+a)\tilde\eta (A-a)\bigr)
+\tfrac12\mathop{\mathrm{Im}}\bigl((A-a)^*\Theta^{2\beta} L \bigr)
\\&\phantom{={}}
+\mathop{\mathrm{Im}}\bigl((A-a)^*\Theta^{2\beta}q_{21}\bigr)
+\mathop{\mathrm{Im}}\bigl((A-a)^*\Theta^{2\beta}q_3\bigr)
\\
&\ge 
\tfrac12(A-a)^*\bigl(\beta\Theta'-C_1r^{-1-2\delta}\Theta\bigr)\Theta^{2\beta-1}(A-a)
\\&\phantom{={}}
+\tfrac12\mathop{\mathrm{Im}}\bigl(A\Theta^{2\beta}L \bigr)
-\tfrac12\mathop{\mathrm{Im}}\bigl(a^*\Theta^{2\beta}L \bigr)
-C_1\Gamma r^{-\rho}\Theta^{2\beta}
-C_1Q
,
\end{split}\label{eq:14.9.4.18.22ff}
\end{align}
where 
\begin{align*}
Q=r^{-1-\min\{\rho,\tau\}+2\delta}\Theta^{2\beta}
+p_i^*r^{-1-\min\{\rho,\tau\}+2\delta}\Theta^{2\beta}g^{ij}p_j.
\end{align*}

We further estimate the terms on the right-hand side of \eqref{eq:14.9.4.18.22ff}.
By Lemma~\ref{lem:15.2.18.14.28} the first term of  \eqref{eq:14.9.4.18.22ff}
can be bounded as 
\begin{align}
\begin{split}
&\tfrac12(A-a)^*\bigl(\beta\Theta'-C_1r^{-1-2\delta}\Theta\bigr)\Theta^{2\beta-1}(A-a)
\\&
\ge 
c_1(A-a)^*\Theta'\Theta^{2\beta-1}(A-a)-C_2Q.
\end{split}\label{eq:14.9.17.12.16ff}
\end{align}
Reusing parts of computations in \eqref{eq:14.9.4.18.22ffaabbqqe}, 
\eqref{eq:14.9.17.12.15ffaabccbqqe} and \eqref{eq:1509101044},
we can write and bound the second term of \eqref{eq:14.9.4.18.22ff} as,
for any $\epsilon\in (0,1)$,
\begin{align}
\begin{split}
\tfrac12\mathop{\mathrm{Im}}\bigl(A\Theta^{2\beta} L \bigr)
&
=\tfrac14\bigl[p_i^*\Theta^{2\beta}\ell^{ij}p_j,\mathrm i A\bigr]
+\tfrac12\mathop{\mathrm{Re}}\bigl(A (1-\eta)(\Theta^{2\beta})'p^r\bigr)
\\&
=
\tfrac14p_i^*\Theta^{2\beta}\Bigl(2 (\nabla^2 r)^{ik}
+(\nabla^r\tilde\eta)(\mathrm dr\otimes\mathrm dr)^{ik}\Bigr)p_k
\\&\phantom{{}={}}
-\tfrac14p_i^*|\mathrm dr|^2(\Theta^{2\beta})'\ell^{ik}p_k
-\tfrac14\mathop{\mathrm{Im}}\bigl(\Theta^{2\beta} (\nabla_i\Delta r)\ell^{ij}p_i\bigr)
\\&\phantom{{}={}}
+\tfrac12A (1-\eta)(\Theta^{2\beta})'A
-\tfrac18(\nabla^r \eta)(\Delta r)(\Theta^{2\beta})'
\\&\phantom{{}={}}
+\tfrac18(1-\eta)|\mathrm dr|^2(\Delta r)(\Theta^{2\beta})''
\\&
\ge
\tfrac12p_i^*\Theta^{2\beta} h^{ik}p_k
-\tfrac12\beta p_i^*|\mathrm dr|^2\Theta'\Theta^{2\beta-1}\ell^{ik}p_k
\\&\phantom{{}={}}
-\epsilon p_i^*r^{-1}\Theta^{2\beta}\ell^{ij}p_j
-\epsilon^{-1}C_3Q.
\end{split}\label{eq:1510131}
\end{align}
As for the third term of \eqref{eq:14.9.4.18.22ff}
use \eqref{eq:14.12.23.13.49},
Lemma~\ref{lem:13.9.2.7.18} and
the Cauchy--Schwarz inequality, and then we obtain for the 
same $\epsilon\in (0,1)$ as above
\begin{align}
\begin{split}
&
-\tfrac12\mathop{\mathrm{Im}}\bigl(a^*\Theta^{2\beta}L\bigr)
\\
&=
\tfrac12p_i^*(\mathop{\mathrm{Im}}a)\Theta^{2\beta} \ell^{ij}p_j
-\tfrac12\mathop{\mathrm{Re}}\bigl((\nabla_ia)^*\Theta^{2\beta}\ell^{ij}p_j\bigr)
-\beta\mathop{\mathrm{Re}}\bigl((1-\eta)a^*\Theta'\Theta^{2\beta-1}p^r\bigr)
\\
&\ge 
-\epsilon p_i^*r^{-1}\Theta^{2\beta}\ell^{ij}p_j
-\epsilon^{-1}C_4Q
.
\end{split}
\label{eq:14.9.17.12.21ff}
\end{align}
As for the fourth term of \eqref{eq:14.9.4.18.22ff},
we have by the Cauchy--Schwarz inequality
\begin{align}
\begin{split}
-C_1\Gamma r^{-\rho}\Theta^{2\beta}
&= 
-\tfrac12C_1[H,\mathrm ir^{-\rho}\Theta^{2\beta}]
+C_1\mathop{\mathrm{Im}}\bigl(r^{-\rho}\Theta^{2\beta}(H-z)\bigr)
\\
&\ge 
-C_5Q
+C_1\mathop{\mathrm{Im}}\bigl(r^{-\rho}\Theta^{2\beta}(H-z)\bigr)
.
\end{split}
\label{eq:14.10.3.2.17ff}
\end{align}
Now we substitute the bounds 
\eqref{eq:14.9.17.12.16ff},
\eqref{eq:1510131}, 
\eqref{eq:14.9.17.12.21ff},
and 
\eqref{eq:14.10.3.2.17ff}
 into 
\eqref{eq:14.9.4.18.22ff},
and obtain
\begin{align}
\begin{split}
\mathop{\mathrm{Im}}\bigl((A-a)^*\Theta^{2\beta}(H-z)\bigr)
&\ge
c_1(A-a)^*\Theta'\Theta^{2\beta-1}(A-a)
\\&\phantom{={}}
+\tfrac12p_i^*\Theta^{2\beta-1}\Bigl(\Theta h^{ij}-4\epsilon r^{-1}\Theta\ell^{ij}-\beta|\mathrm dr|^2\Theta'\ell^{ij}\Bigr)p_j
\\&\phantom{={}}
-\epsilon^{-1}C_6Q
+C_1\mathop{\mathrm{Im}}\bigl(r^{-\rho}\Theta^{2\beta}(H-z)\bigr).
\end{split}\label{eq:14.9.4.18.23ff}
\end{align}

If we choose $\epsilon>0$ small enough, we have for the second term 
of \eqref{eq:14.9.4.18.23ff}
\begin{align}
\tfrac12p_i^*\Theta^{2\beta-1}\Bigl(h^{ij}-4\epsilon r^{-1}\Theta\ell^{ij}-\beta|\mathrm dr|^2\Theta'\ell^{ij}\Bigr)p_j
\ge c_2p_i^*\Theta^{2\beta}h^{ij}p_j.
\label{1510132}
\end{align}
Hence it finally remains to bound $-Q$ below,
but by \eqref{eq:13.9.7.9.32}
we can estimate it as 
\begin{align}
\begin{split}
-Q
&\ge 
-C_7r^{-1-\min\{\rho,\tau\}+2\delta}\Theta^{2\beta}
-2\mathop{\mathrm{Re}}
\bigl(r^{-1-\min\{\rho,\tau\}+2\delta}\Theta^{2\beta}(H-z)\bigr)
.
\end{split}
\label{eq:14.9.17.12.17ff}
\end{align}
By \eqref{eq:14.9.4.18.23ff}, \eqref{1510132} and \eqref{eq:14.9.17.12.17ff}, if we 
set 
\begin{align*}
\gamma=\mathrm iC_1r^{-\rho}+2\epsilon^{-1}C_6r^{-1-\min\{\rho,\tau\}+2\delta},
\end{align*}
then the assertion follows.
\end{proof}

\subsection{Applications}\label{subsec:15.3.9.14.28}
Now we are going to prove  Theorem~\ref{prop:radiation-conditions}
and Corollaries~\ref{thm:12.7.2.7.9b}--\ref{thm:13.9.9.8.23} in this order.

\subsubsection{Radiation condition bounds for complex spectral parameters}

\begin{proof}[Proof of Theorem~\ref{prop:radiation-conditions}]
Let $I\subseteq \mathcal I$ be a compact interval.
For $\beta=0$ the assertion is obvious by Theorem~\ref{thm:12.7.2.7.9},
and hence we may let $\beta\in (0,\beta_c)$.
We take any  
\begin{align*}
\delta\in \bigl(0,\min\{\rho'/2,\rho/4,\tau/4\}\bigr]
\cap \bigl(0,\min\{\rho/2,\tau/2\}-\beta\bigr).
\end{align*}
By Lemma~\ref{lem:14.10.22.18.8ff}, the Cauchy--Schwarz inequality 
and  Theorem~\ref{thm:12.7.2.7.9}
there exists $C_1>0$ such that for any state
$\phi=R(z) \psi$ with $\psi\in C^\infty_{\mathrm c}(M)$ and $z\in I_+$
\begin{align}
\begin{split}
&\bigl\|\Theta'{}^{1/2}\Theta^{\beta-1/2}(A-a)\phi\bigr\|^2
+\bigl\langle p_i^*\Theta^{2\beta}h^{ij}p_j\bigr\rangle_\phi\\
&\le 
C_1\Bigl[\|\Theta^\beta(A-a)\phi\|_{B^*}\|\Theta^\beta\psi\|_B
+
\bigl\|r^{-(1+\min\{\rho,\tau\})/2+\delta}\Theta^\beta\phi\bigr\|^2
\\&\phantom{=C_1\Bigl[}
+\bigl\|r^{(1-\min\{\rho,\tau\})/2+\delta}\Theta^\beta\psi\bigr\|^2\Bigr]\\
&\le C_2R_\nu^{-2\beta}\Bigl[\|r^\beta(A-a)\phi\|_{B^*}\|r^\beta\psi\|_B
+\|r^\beta\psi\|_B^2\Bigr].
\end{split}
\label{eq:14.10.4.1.18}
\end{align}
Here we note that $r^\beta(A-a)\phi\in B^*$ for each $z\in I_+$ and hence the quantity
on the right-hand side of \eqref{eq:14.10.4.1.18} is finite.
In fact, this can be verified by commuting $R(z)$ and powers of $r$ sufficiently many times 
and using the fact that $\psi\in C^\infty_{\mathrm c}(M)$.
Then by \eqref{eq:14.10.4.1.18} it follows 
\begin{align}
\begin{split}
&R_\nu^{2\beta}\bigl\|\Theta'{}^{1/2}\Theta^{\beta-1/2}(A-a)\phi\bigr\|^2
+R_\nu^{2\beta}\bigl\langle p_i^*\Theta^{2\beta}h^{ij}p_j\bigr\rangle_\phi\\
&\le
C_2\Bigl[\|r^\beta(A-a)\phi\|_{B^*}\|r^\beta\psi\|_B
+\|r^\beta\psi\|_B^2\Bigr].
\end{split}
\label{eq:14.10.4.1.19}
\end{align}
In the first term on the left-hand side of \eqref{eq:14.10.4.1.19}
we take the
supremum in $\nu\ge 0$ noting \eqref{eq:15.2.15.5.9},
and then obtain
\begin{align*}
\begin{split}
c_1\|r^\beta(A-a)\phi\|_{B^*}^2
\le
C_2\Bigl[\|r^\beta(A-a)\phi\|_{B^*}\|r^\beta\psi\|_B
+\|r^\beta\psi\|_B^2\Bigr],
\end{split}
%\label{eq:14.10.4.1.20}
\end{align*}
which implies
\begin{align}
\begin{split}
\|r^\beta(A-a)\phi\|_{B^*}^2
\le
C_3\|r^\beta\psi\|_B^2.
\end{split}
\label{eq:14.10.4.1.20}
\end{align}
As for the second term on the left-hand side of
\eqref{eq:14.10.4.1.19} we use \eqref{eq:14.10.4.1.20}, the concavity
of $\Theta$ and Lebesgue's monotone convergence theorem 
and then obtain by letting $\nu\to\infty$
\begin{align*}
\begin{split}
\langle p_i^*r^{2\beta}h^{ij}p_j\rangle_\phi
\le
C_4\|r^\beta\psi\|_B^2.
\end{split}
\end{align*}
Hence we are done.
\end{proof}

\subsubsection{Limiting absorption principle}

\begin{proof}[Proof of Corollary~\ref{thm:12.7.2.7.9b}]
Let $s>1/2$ and $\epsilon\in (0,\min\{s-1/2,\beta_c,(2+\rho)/4\})$ be as in the assertion. 
Let $s'=s-\epsilon$.
We decompose for $n\ge 0$ and $z, z'\in I_+$
\begin{align}
\begin{split}
R(z)-R(z')
&=
\chi_nR(z)\chi_n
-\chi_nR(z')\chi_n
\\&\phantom{{}=}
+\bigl(R(z)-\chi_nR(z)\chi_n\bigr)
-\bigl(R(z')-\chi_nR(z')\chi_n\bigr).
\end{split}\label{eq:14.12.30.20.42}
\end{align}
We estimate the last two terms of \eqref{eq:14.12.30.20.42} as follows:
By Theorem~\ref{thm:12.7.2.7.9}
we have uniformly in $n\ge 0$ and $z, z'\in I_+$
\begin{align}
\begin{split}
&
\|R(z)-\chi_nR(z)\chi_n\|_{\mathcal B(\mathcal H_s,\mathcal H_{-s})}
\\&
\le 
\|r^{-s}R(z)\bar\chi_nr^{-s}\|_{\mathcal B(\mathcal H)}
+\|r^{-s}\bar\chi_nR(z)\chi_nr^{-s}\|_{\mathcal B(\mathcal H)}
\\&
\le C_1 R_n^{s'-s}=C_1R_n^{-\epsilon},
\end{split}
\end{align}
and, similarly,
\begin{align}
\|R(z')-\chi_nR(z')\chi_n\|_{\mathcal B(\mathcal H_s,\mathcal H_{-s})}
\le C_2 R_n^{-\epsilon}.
\end{align}
As for the first and second terms on the right-hand side of 
\eqref{eq:14.12.30.20.42},
using the expressions \eqref{eq:12.6.27.0.43} and 
\begin{align}
\mathrm{i}[H,\chi_{n+1}]
=\mathop{\mathrm{Re}}(\chi_{n+1}'p^r)
=\mathop{\mathrm{Re}}(\chi_{n+1}'A)
\label{eq:1}
\end{align}
and noting the identity $\overline{a_{\bar z}}=a_z$,
we write for $n\ge 0$
\begin{align}
\begin{split}
&\chi_nR(z)\chi_n-\chi_nR(z')\chi_n
\\&=\chi_nR(z)\bigl\{\chi_{n+1}(H-z')-(H-z)\chi_{n+1}\bigr\}R(z')\chi_n
\\&
=
\tfrac{\mathrm i}2\chi_nR(z)\chi_{n+1}'(A-a_{z'})R(z')\chi_n
+\tfrac{\mathrm i}2\chi_nR(z)(A+a_{\bar z})^*\chi_{n+1}'R(z')\chi_n.
\\&\phantom{{}=}
-\tfrac{\mathrm i}2\chi_nR(z)(a_z-a_{z'})\chi_{n+1}'R(z')\chi_n
-(z-z')\chi_nR(z)\chi_{m}R(z')\chi_n
\\&\phantom{{}=}
-(z-z')\chi_nR(z)\chi_{m,n+1}(a_z+a_{z'})^{-1}(A-a_{z'})R(z')\chi_n
\\&\phantom{{}=}
+(z-z')\chi_nR(z)(A+a_{\bar z})^*\chi_{m,n+1}(a_z+a_{z'})^{-1}R(z')\chi_n
\\&\phantom{{}=}
-(z-z')\chi_nR(z)\bigl[A,\chi_{m,n+1}(a_z+a_{z'})^{-1}\bigr]R(z')\chi_n.
\end{split}
\label{eq:190531}
\end{align}
Here $m\ge 0$ is fixed so that $(a_z+a_{z'})^{-1}$ is not singular on $\mathop{\mathrm{supp}}\bar\chi_m$.
Then by Theorems~\ref{thm:12.7.2.7.9} and \ref{prop:radiation-conditions}
we have uniformly in $n\ge 0$ and $z,z'\in I_+$
\begin{align}
\|\chi_nR(z)\chi_n-\chi_nR(z')\chi_n\|_{\mathcal B(\mathcal H_s,\mathcal H_{-s})}
\le C_3R_n^{-\epsilon}+C_4R_n^{\max\{0,1-\epsilon\}}|z-z'|.
\label{eq:14.12.30.20.43}
\end{align}
Note that the upper bound $\epsilon<(2+\rho)/4$ is used to control the 
last term of \eqref{eq:190531}.

Summing up \eqref{eq:14.12.30.20.42}--\eqref{eq:14.12.30.20.43},
we obtain uniformly in $n\ge 0$ and $z,z'\in I_+$ 
\begin{align*}
\|R(z)-R(z')\|_{\mathcal B(\mathcal H_s,\mathcal H_{-s})}
=C_5 R_n^{-\epsilon}+C_5R_n^{\max\{0,1-\epsilon\}}|z-z'|.
\end{align*}
Now, if $|z-z'|\le 1$, then we choose $R_n\le |z-z'|^{-1}<2R_n$,
and then obtain 
\begin{align}
\|R(z)-R(z')\|_{\mathcal B(\mathcal H_s,\mathcal H_{-s})}
\le C_6|z-z'|^{\min\{\epsilon,1\}}.
\label{eq:ho}
\end{align}
The same bound is trivial for $|z-z'|\ge 1$, and hence 
the H\"older continuity \eqref{eq:14.12.30.21.52} for $R(z)$ follows from
\eqref{eq:ho}.  The H\"older continuity \eqref{eq:14.12.30.21.52} for
$pR(z)$ follows by using in addition \eqref{eq:13.9.7.9.32} or
alternatively the first resolvent equation.

The existence of the limits of \eqref{eq:14.12.30.21.53}
is an immediate consequence of \eqref{eq:14.12.30.21.52}.
By Theorem~\ref{thm:12.7.2.7.9} the limits $p^\alpha R(\lambda+\mathrm i0)$
actually map into $B^*$,
and moreover they extend  continuously to  maps $B\to B^*$ by a density argument.
Finally, since $R(z)$ for $z\in I_+$ maps into $\mathcal N$, 
it follows by \eqref{eq:14.12.30.21.52} and approximation
arguments 
that $R(\lambda\pm\mathrm i0)$ map into $\mathcal N$.
Hence we are done.
\end{proof}

\subsubsection{Radiation condition bounds for real spectral parameters}
\label{subsec:Application}

\begin{proof}[Proof of Corollary~\ref{thm:radiation-conditions}]
 The corollary 
follows from  Theorem~\ref{prop:radiation-conditions},
Corollary~\ref{thm:12.7.2.7.9b} 
and approximation
arguments. Note the elementary property 
\begin{align*}
  \|\psi\|_{B^*}=\sup_{n\geq 0} \|\chi_n \psi\|_{B^*};\quad \psi\in B^*.
\end{align*}
Hence we are done.
\end{proof}

\subsubsection{Sommerfeld uniqueness result}

\begin{proof}[Proof of Corollary~\ref{thm:13.9.9.8.23}]
Let $\lambda\in\mathcal I$, 
$\phi\in \mathcal H_{\mathrm{loc}}$ and $\psi\in r^{-\beta}B$ with
$\beta\in [0,\beta_c)$.
We first assume $\phi=R(\lambda+\mathrm i0)\psi$.
Then \ref{item:13.7.29.0.29} and \ref{item:13.7.29.0.28} of 
the corollary obviously hold by
Corollaries~\ref{thm:12.7.2.7.9b} and \ref{thm:radiation-conditions}.
Conversely, assume \ref{item:13.7.29.0.29} and \ref{item:13.7.29.0.28}
of the corollary, and let 
\begin{align*}
\phi'=\phi-R(\lambda+\mathrm i0)\psi.
\end{align*}
Then by Corollaries~\ref{thm:12.7.2.7.9b} and \ref{thm:radiation-conditions}
it follows that 
$\phi'$ 
satisfies \ref{item:13.7.29.0.29} and \ref{item:13.7.29.0.28} 
of the corollary with $\psi=0$.
In addition, we can verify $\phi'\in B^*_0$ by the virial-type argument.
In fact noting the identity
\begin{align*}
2\mathop{\mathrm{Im}}\bigl(\chi_\nu(H-\lambda)\bigr)
=(\mathop{\mathrm{Re}}a)\chi'_\nu+\mathop{\mathrm{Re}}\bigl(\chi'_\nu(A-a)\bigr),
\end{align*}
cf.\ \eqref{eq:12.6.27.0.43} and \eqref{eq:1},
we conclude that 
\begin{align}
0\le \bigl\langle (\mathop{\mathrm{Re}}a)\bar\chi'_\nu\bigr\rangle_{\phi'}
\le 
\mathop{\mathrm{Re}}\bigl\langle \chi'_\nu(A-a)\bigr\rangle_{\phi'}.
\label{eq:13.9.8.10.55}
\end{align}
Taking the limit $\nu\to \infty$ and using $\phi'\in r^\beta B^*$ and 
$(A-a)\phi'\in r^{-\beta}B^*_0$ 
in (\ref{eq:13.9.8.10.55}),
indeed 
we obtain $\phi\in B^*_0$.
 By Theorem~\ref{thm:13.6.20.0.10} 
it follows $\phi'=0$, and hence $\phi=R(\lambda+\mathrm i0)\psi$.
\end{proof}

\end{document}